\documentclass[letterpaper, 10 pt, conference]{ieeeconf}%
\usepackage{graphicx}
\usepackage{epsfig}
\usepackage{times}
\usepackage{amsmath}
\usepackage{thmtools,thm-restate}
\usepackage{amssymb}
\usepackage[section]{placeins}
\usepackage{placeins}
\usepackage{amsmath}
\usepackage{multirow}
\usepackage{graphicx}
\usepackage[normalem]{ulem}
\usepackage{subcaption}
\usepackage{algorithm,tabularx}
\usepackage{algpseudocode}
\usepackage{amsfonts}
\usepackage{cases}
\usepackage{romannum}
\usepackage{textcomp}
\usepackage{xcolor}%
\usepackage{textcomp}
\providecommand{\U}[1]{\protect\rule{.1in}{.1in}}
\providecommand{\U}[1]{\protect\rule{.1in}{.1in}}

\IEEEoverridecommandlockouts
\newtheorem{assumption}{Assumption}
\newtheorem{theorem}{Theorem}

\newtheorem{lemma}{Lemma}
\newtheorem{proposition}{Proposition}
\newtheorem{remark}{Remark}
\newtheorem{definition}{Definition}

\useunder{\uline}{\ul}{}
\makeatletter
\newcommand{\multiline}[1]{  \begin{tabularx}{\dimexpr\linewidth-\ALG@thistlm}[t]{@{}X@{}}
#1
\end{tabularx}
}
\makeatother
\usepackage{cite}

\usepackage{picins}
\usepackage{centernot}
\let\labelindent\relax
\usepackage{enumitem}
\setlist[itemize]{leftmargin=*}
\usepackage{makecell}
\usepackage[normalem]{ulem}
\usepackage{cuted}
\usepackage{hyperref}
\usepackage{stfloats}
\usepackage{lscape}

\newcommand{\R}{\mathbb{R}}
\newcommand{\N}{\mathbb{N}}

\newcommand{\T}{\top}
\newcommand{\I}{\mathbf{I}}
\newcommand{\0}{\mathbf{0}}

\newcommand{\X}{\mathcal{X}}
\newcommand{\diag}{\text{diag}}
\newcommand{\tsup}[1]{\textsuperscript{#1}}
\newcommand{\mb}[1]{\mathbf{#1}}
\renewcommand{\H}{\mathcal{H}}

\newcommand{\bm}[1]{\begin{bmatrix}#1\end{bmatrix}}

\algdef{SE}[SUBALG]{Indent}{EndIndent}{}{\algorithmicend\ }%
\algtext*{Indent}
\algtext*{EndIndent}

\title{\LARGE \bf
Centralized and Decentralized Techniques for Analysis and Synthesis of Non-Linear Networked Systems
\vspace{-15pt}
}

\author{Shirantha Welikala, Hai Lin and Panos J. Antsaklis 
\thanks{The support of the National Science Foundation (Grant No. IIS-1724070, CNS-1830335, IIS-2007949) is gratefully acknowledged.}
\thanks{The authors are with the Department of Electrical Engineering, College of Engineering, University of Notre Dame, IN 46556, \texttt{{\small \{wwelikal,hlin1,pantsakl\}@nd.edu}}.}}

\begin{document}

\maketitle
\thispagestyle{empty}
\pagestyle{empty}


\begin{abstract}
In this paper, we develop centralized and decentralized techniques for analyzing and synthesizing networked systems comprised of interconnected sets of non-linear subsystems - only using the subsystem dissipativity properties. In particular, this paper extends a recent work that proposed dissipativity based centralized analysis techniques for non-linear networked systems as linear matrix inequality (LMI) problems. First, we consider four networked system configurations (NSCs) of interest and provide centralized stability/dissipativity \emph{analysis} techniques for them as LMI problems. Second, we show that the centralized interconnection topology \emph{synthesis} techniques for these NSCs can also be developed as LMI problems. This enables synthesizing the interconnection topology among the subsystems so that specific stability/dissipativity measures of the networked system are optimized. Next, we show that the proposed analysis and synthesis techniques can be implemented in a decentralized and compositional manner (i.e., subsystems can be added/removed conveniently). Finally, we include several numerical results to demonstrate our contributions.
\end{abstract}

\section{Introduction}\label{Sec:Introduction}

Large-scale networked systems have gained a renewed attention over the past several years due to their broad applications in infrastructure networks \cite{Agarwal2021,Arcak2022,Tang2021}, biological networks \cite{Chaves2018,Rufino2012,Arcak2008}, vehicle platooning \cite{Song2022Ax,Karafyllis2021,Antonelli2013}, electronic circuits \cite{Jafarian2019,Jeltsema2005}, mechanical networks \cite{Zhu2015,Papageorgiou2004} and so on. Related to such networked systems, the main research thrusts have been focused on addressing problems such as analysis/verification \cite{Ebihara2017,Arcak2016,Cremean2003}, optimization \cite{Welikala2020P7,Oltafi-Saber2004,Welikala2020J2}, abstraction \cite{Lavaei2020,Zamani2018,Lavaei2020b}, controller synthesis \cite{WelikalaP32022,Nejati2022,Morrison2021} and network topology synthesis \cite{Ghosh2021,Ebihara2017,Rufino2018,Ghanbari2016}. Moreover, the literature on large-scale networked systems can also be categorized based on the nature of the constituent subsystems as: discrete \cite{Welikala2022Ax2}, continuous \cite{Agarwal2021}, linear \cite{WelikalaP32022}, piece-wise linear \cite{Welikala2020J2}, non-linear \cite{Arcak2022}, stochastic \cite{Lavaei2020}, switched \cite{Lavaei2022} and hybrid \cite{Nejati2022}. 

With respect to the literature above, this paper aims to address the research problems: analysis/verification and network topology synthesis of large-scale networked systems comprised of non-linear subsystems (henceforth referred to as \emph{analysis} and \emph{synthesis} of non-linear networked systems). In particular, we are interested in \emph{analyzing} the stability/dissipativity properties of a given non-linear networked system and, when necessary, \emph{synthesizing} the interconnection topology of the non-linear networked system to guarantee/optimize such stability/dissipativity properties. 

To achieve these goals, we only assume the knowledge of the involved subsystem dissipativity properties such as their passivity, passivity indices, $L_2$-stability and $L_2$-gain values. Note that identifying such dissipativity properties is far more convenient, efficient and reliable than having to identify entire dynamic models \cite{Koch2021,WelikalaP42022,Arcak2022,Xia2014,Zakeri2021}. Further, we formulate the interested non-linear networked system analysis and synthesis problems as linear matrix inequality (LMI) problems. Thus, they can be conveniently implemented and efficiently solved \cite{Boyd1994,Lofberg2004}. Furthermore, we formulate these LMI problems so that they can be implemented in a decentralized and compositional manner. Consequently, new subsystems can be added/removed to/from a networked system conveniently without having to re-design the existing interconnections \cite{Agarwal2021,WelikalaP32022}. Due to these unique qualities, the proposed solutions in this paper can be used for higher-level design of large-scale networked systems. To summarize, we provide dissipativity-based centralized and decentralized LMI techniques to analyze and synthesize non-linear networked systems.

This paper was inspired by the well-known work \cite{Arcak2016,Arcak2022} where dissipativity based centralized LMI techniques have been developed to analyze non-linear networked systems. Compared to \cite{Arcak2016,Arcak2022}, we consider two special non-linear networked system configurations (NSCs) with multiple applications. Moreover, as mentioned earlier, we not only provide techniques for analysis but also address the corresponding synthesis and decentralization problems.  

It should be noted that even though the analysis techniques proposed in \cite{Arcak2016,Arcak2022} have been derived in a compositional manner, they can only be evaluated in a centralized setting. To address this challenge, inspired by Sylvester's criterion \cite{Antsaklis2006}, the work in \cite{Agarwal2021,WelikalaP32022} propose a decentralized and compositional iterative process to analyze linear networked systems. In this paper, we adopt this idea to analyze non-linear networked systems in a decentralized and compositional manner.

A critical difference between this work and \cite{Agarwal2021,WelikalaP32022} is that we only use dissipativity properties of the subsystems, whereas \cite{Agarwal2021,WelikalaP32022} use subsystem dynamic models. Note, however, that this additional information and the linear dynamics assumed in \cite{Agarwal2021,WelikalaP32022} also enable the synthesis of standard feedback controllers and observers for linear networked systems. In contrast, due to the said complexities that we consider, here we do not aim to synthesize controllers. Instead, we aim to synthesize appropriate interconnection topologies (in a centralized or decentralized and compositional manner) for non-linear networked systems when the analysis fails.  

Nevertheless, there is a wealth of literature that focuses on the synthesis of controllers for networked systems. Some recent examples are as follows. Considering a linear networked system, a decentralized observer based controller is synthesized in \cite{Elmahdi2015}. For stochastic hybrid networked systems, to enforce a particular class of global specifications, a compositional framework is proposed in \cite{Nejati2022} to construct local control barrier functions. This solution is extended in \cite{Jahanshahi2022} for discrete-time partially observable stochastic networked systems. For non-linear dynamical systems, a feed-forward control procedure is proposed in \cite{Morrison2021} using model reduction and bifurcation theory. A decentralized controller synthesis approach is proposed in \cite{Liu2019b} for discrete-time linear networked systems exploiting approximations of robust controlled invariant sets \cite{Rungger2017}. Compositionally constructed abstractions of networked systems are used to synthesize controllers for discrete-time stochastic \cite{Lavaei2019,Lavaei2022,Lavaei2017,Lavaei2020}, continuous-time hybrid \cite{Nejati2021,Awan2020} and nonlinear \cite{Zamani2018} networked systems.

As mentioned earlier, these controller synthesis techniques require additional information and assumptions regarding the involved subsystems and their interconnection topology. In contrast, in this paper, we only use subsystem dissipativity properties and focus on synthesizing the interconnection topology - which describes how each subsystem output is connected to the other subsystem inputs in the considered networked system (characterized by a block matrix called the ``interconnection matrix''). Note that this approach is more reasonable when there are ready-made standard controllers (perhaps with tunable parameters that may change their dissipativity properties). For example, see the scenarios considered in \cite{Cremean2003,Ghanbari2016,Ebihara2017,Rufino2018,Arcak2022,Xia2014,Zakeri2019}.

However, literature focusing directly on synthesizing this interconnection matrix of networked systems are few and far between (especially for non-linear networked systems). In fact, the leading paper that motivated us to address this particular synthesis problem is \cite{Xia2014} (see also \cite{Xia2018}). In particular, \cite{Xia2014,Xia2018} consider a networked system comprised only of one or two subsystems and propose a set of non-linear inequalities to guide the synthesis of the corresponding interconnection matrix. In contrast, we consider a networked system with an arbitrary number of subsystems and provide LMI conditions to efficiently and conveniently synthesize the interconnection matrix (centrally or decentrally). Therefore, the proposed synthesis results in this paper can be seen as a generalization of \cite{Xia2014,Xia2018}. Note also that the techniques proposed in \cite{Xia2014,Xia2018} have influenced several recent works on problems such as designing switched controllers \cite{Ghanbari2017}, networked controllers \cite{Rahnama2018} and adaptive controllers \cite{Zakeri2019}. Therefore, it is reasonable to expect that the proposed techniques in this paper may also be able to generalize the techniques proposed in \cite{Ghanbari2017,Rahnama2018,Zakeri2019}.

In addition, the work in \cite{Ebihara2017} has considered the interconnection topology synthesis problem limited to linear and positive networked systems. However, it requires the explicit knowledge of the involved linear positive subsystems. In \cite{Rufino2018}, symmetries in the interconnection matrix are exploited to reduce the computational complexity of the networked system analysis. However, this approach only allows the interconnection matrix to be diagonally scaled from a predefined value (to recover symmetries). For stability analysis of large-scale interconnected systems, similar symmetry based techniques have been used in \cite{Ghanbari2016,Goodwine2013}. However, they assume the interconnection matrix (hence the topology) as a given - rather than treating it as a decision variable. Taking a graph theoretic approach, necessary conditions for the stability of a class of non-linear networked systems are given in terms of the interconnection matrix in \cite{Cremean2003}. While this solution leads to identify several types of interconnection topologies that guarantee stability, it is only applicable to a small class of non-linear systems with known dynamics.

\paragraph{\textbf{Contributions}}
With respect to the literature mentioned above, our contributions can be summarized as follows: 
(1) We consider several networked system configurations (NSCs) that can be used to model several exciting and widely used non-linear networked system configurations;
(2) For the centralized analysis (stability/dissipativity) of such NSCs, we propose linear matrix inequality (LMI) based techniques;
(3) For the centralized synthesis (interconnection topology) of such NSCs, we propose LMI based approaches;
(4) We propose decentralized and compositional (i.e., resilient to subsystem additions and removals) counterparts for the proposed centralized analysis and synthesis approaches; and 
(5) We provide several numerical results to support our theoretical results.

\paragraph{\textbf{Organization}} 
This paper is organized as follows. In Section \ref{Sec:Preliminaries}, we provide several preliminary concepts and results related to dissipativity, network matrices and their positive definiteness analysis. Different networked system configurations and applicable centralized techniques to analyze their stability and dissipativity are discussed in Sec. \ref{Sec:NetworkedSystems}. In Sec. \ref{Sec:CentralizedSynthesis}, we propose centralized techniques to synthesize the interconnection matrices involved in the considered networked systems. Section  \ref{Sec:DecentralizedAnalysisAndSynthesis} provide the decentralized counterparts of the proposed centralized analysis and centralized synthesis techniques. Finally, Sec. \ref{Sec:NumericalResults} discusses several numerical examples before concluding the paper in Sec. \ref{Sec:Conclusion}.

\paragraph{\textbf{Notation}}
The sets of real and natural numbers are denoted by $\R$ and $\N$, respectively. We define $\N_N\triangleq\{1,2,\ldots,N\}$ where $N\in\N$. 
An $n\times m$ block matrix $A$ can be represented as $A=[A_{ij}]_{i\in\N_n, j\in\N_m}$ where $A_{ij}$ is the $(i,j)$\tsup{th} block of $A$. $[A_{ij}]_{j\in \N_m}$ and $\diag(A_{ii}:i\in\N_n)$ represent a block row matrix and a block diagonal matrix, respectively. We also define $\{A_i\} \triangleq \{A_{ii}\}\cup\{A_{ij},j\in\N_{i-1}\}\cup\{A_{ji}:j\in\N_i\}$. If $\Psi\triangleq[\Psi^{kl}]_{k,l \in \N_m}$, its block element-wise form \cite{WelikalaP32022} is denoted as $\mbox{BEW}_n(\Psi) \triangleq [[\Psi^{kl}_{ij}]_{k,l\in\N_m}]_{i,j\in\N_n}$ (note that, when indexing, subscripts and superscripts are used interchangeably, e.g., $A^{ij} \equiv A_{ij}$). The transpose of a matrix $A$ is denoted by $A^\T$ and $(A^\T)^{-1} = A^{-\T}$. The zero and identity matrices are denoted by $\0$ and $\I$, respectively (dimensions will be clear form the context). A symmetric positive definite (semi-definite) matrix $A\in\R^{n\times n}$ is represented as $A=A^\T>0$ ($A=A^\T \geq 0$). Unless stated otherwise, $A>0 \iff A=A^\T>0$. The symbol $\star$ is used to represent redundant conjugate matrices (e.g., 
$\scriptsize \bm{A & B \\ \star & C} \equiv \bm{A & B \\ B^\T & C}$, $\scriptsize \bm{A_{ij} & B_{ij} \\ \star & C_{ij}} \equiv \bm{A_{ij} & B_{ij} \\ B_{ji}^\T & C_{ij}}$ and $A^\T B\, \star \equiv A^\T B A$). 
The symmetric part of a matrix $A$ is defined as $\H(A) \triangleq A+A^\T$ and $\H(A_{ij}) \triangleq A_{ij}+A_{ji}^\T$. $\mb{1}_{\{\cdot\}}$ represents the indicator function and $e_{ij} \triangleq \I \cdot \mb{1}_{\{i=j\}}$.

\section{Preliminaries}\label{Sec:Preliminaries}

\subsection{Equilibrium Independent Dissipativity}

Consider the non-linear dynamical system  
\begin{equation}\label{Eq:GeneralSystem}
\begin{aligned}
    \dot{x}(t) = f(x(t),u(t)),\\
    y(t) = h(x(t),u(t)),
    \end{aligned}
\end{equation}
where $x(t)\in\R^{n}$, $u(t)\in \R^{q}$, $y(t)\in\R^{m}$, and $f:\R^{n}\times \R^{q} \rightarrow \R^{n}$ and $h:\R^{n}\times \R^{q}\rightarrow \R^{m}$ are continuously differentiable. Suppose there exists a set $\X\subset \R^{n}$ where for every $\tilde{x}\in\X$ there is a unique $\tilde{u}\in\R^{q}$ satisfying $f(\tilde{x},\tilde{u})=\0$, and both $\tilde{u}$ and $\tilde{y}\triangleq h(\tilde{x},\tilde{u})$ are implicit functions of $\tilde{x}$. 

Note that when $\X=\{x^*\}$ (a singleton), $x^*$ is the unique equilibrium of \eqref{Eq:GeneralSystem} (and $x^*=\0$ is a conventional assumption that we avoid making here). The \emph{equilibrium-independent dissipativity} (EID) property \cite{Arcak2022} defined next (which includes the conventional dissipativity property \cite{Willems1972a}) enables stability and dissipativity analysis of \eqref{Eq:GeneralSystem} without having explicit knowledge regarding the equilibrium point(s) of \eqref{Eq:GeneralSystem}.  

\begin{definition}
The system \eqref{Eq:GeneralSystem} is equilibrium-independent dissipative (EID) with supply rate $s:\R^{q}\times\R^{m}\rightarrow \R$ if there exists a continuously differentiable storage function $V:\R^{n}\times \X \rightarrow \R$ satisfying 
$\forall (x,\tilde{x},u)\in\R^{n}\times \X \times \R^{q}$:
$V(x,\tilde{x})>0$, $V(\tilde{x},\tilde{x})=0$ and 
\begin{equation}
    \dot{V}(x,\tilde{x}) = \nabla_x V(x,\tilde{x})f(x,u)
    \leq  s(u-\tilde{u},y-\tilde{y}).
\end{equation}
\end{definition}

Note that the above EID concept can be specialized based on the form of the used storage function $s(\cdot,\cdot)$. 
For example, the well-known $(Q,S,R)$-dissipativity property \cite{WelikalaP32022} is defined with respect to a quadratic supply function. An equivalent dissipativity property (defined next) is used in the remainder of this paper for notational convenience.  

\begin{definition}\label{Def:X-EID}
The system \eqref{Eq:GeneralSystem} is $X$-equilibrium independent dissipative ($X$-EID) if it is EID with respect to the quadratic supply rate
\begin{equation}
    s(u-\tilde{u},y-\tilde{y}) \triangleq 
    \bm{u-\tilde{u}\\y-\tilde{y}}^\T 
    \underbrace{\bm{X^{11} & X^{12}\\X^{21} & X^{22}}}_{\triangleq X}
    \bm{u-\tilde{u}\\y-\tilde{y}}.
\end{equation}
\end{definition}

If $\X = \{x^*\}$ where $x^*=\0$ and corresponding equilibrium control and output values are $u^* = \0$ and $y^* = \0$, respectively, $X$-EID and $(Q,S,R)$-dissipativity are equivalent and   
\begin{equation}
X \triangleq \bm{X^{11} & X^{12}\\ X^{21} & X^{22}} \equiv 
\bm{R & S^\T\\ S & Q}.
\end{equation}
Similar to the $(Q,S,R)$-dissipativity, $X$-EID covers several properties of interest based on the choice of $X$.

\begin{remark}\label{Rm:X-DissipativityVersions}
If the system \eqref{Eq:GeneralSystem} is $X$-EID with:  
\begin{enumerate}
    \item $X = \bm{\0 & \frac{1}{2}\I \\ \frac{1}{2}\I & \0}$, then it is \emph{passive};
    \item $X = \bm{-\nu\I & \frac{1}{2}\I \\ \frac{1}{2}\I & -\rho\I}$, then it is \emph{strictly passive} ($\nu$ and $\rho$ are input and output passivity indices, respectively \cite{WelikalaP42022});
    \item $X = \bm{\gamma^2\I & \0 \\ \0 & -\I}$, then it is finite-gain \emph{$L_2$-stable} ($\gamma$ is the $L_2$-gain);
\end{enumerate}
in an equilibrium-independent manner.
\end{remark}

\subsection{Network Matrices}

Consider a directed network $\mathcal{G}_n=(\mathcal{V},\mathcal{E})$ where $\mathcal{V} \triangleq \{\Sigma_i:i\in\N_n\}$ is the set of subsystems (nodes), $\mathcal{E} \subset \mathcal{V}\times \mathcal{V}$ is the set of inter-subsystem interconnections (edges), and $n\in\N$. We next define a class of matrices named ``network matrices'' \cite{WelikalaP32022} that corresponds to a such network topology $\mathcal{G}_n$.

\begin{definition}\cite{WelikalaP32022}\label{Def:NetworkMatrices}
	Given a network $\mathcal{G}_n=(\mathcal{V},\mathcal{E})$, any $n\times n$ block matrix $\Theta = \bm{\Theta_{ij}}_{i,j\in\N_n}$ is a \emph{network matrix} if: 
	(1) any information specific to the subsystem $i$ is embedded only in its $i$\tsup{th} block row or block column, and 
	(2) $(\Sigma_i,\Sigma_j) \not\in \mathcal{E}$ and $(\Sigma_j,\Sigma_i)\not\in \mathcal{E}$ implies $\Theta_{ij}=\Theta_{ji}=\0$ for all $i,j\in\N_n$.
\end{definition}


Based on this definition, any $n \times n$ block matrix $\Theta=[\Theta_{ij}]_{i,j\in\N_n}$ is a network matrix of $\mathcal{G}_n$ if $\Theta_{ij}$ is a coupling weight matrix corresponding to the edge $(\Sigma_i,\Sigma_j)\in\mathcal{V}$. Note also that any $n \times n$ block diagonal matrix $\Theta=\diag(\Theta_{ii}:i\in \N_n)$ will be a network matrix of any network with $n\in\N$ subsystems if $\Theta_{ii}$ is specific only to the subsystem $i$.
The following lemma provides several useful properties of such network matrices - first established in \cite{WelikalaP32022}.

\begin{lemma}
	\label{Lm:NetworkMatrixProperties}
	\cite{WelikalaP32022}
	Given a network $\mathcal{G}_n$, a few corresponding block network matrices $\Theta,\Phi,\{\Psi^{kl}:k,l\in\N_m\}$, and some arbitrary block-block matrix $\Psi\triangleq[\Psi^{kl}]_{k,l \in \N_m}$:
	\begin{enumerate}
		\item $\Theta^\T$, \ $\alpha \Theta + \beta \Phi$ are network matrices for any $\alpha,\beta \in \R$.
		\item $\Phi \Theta$, $\Theta\Phi$ are network matrices whenever $\Phi$ is a block diagonal network matrix.
		\item $\mbox{BEW}(\Psi)\triangleq [[\Psi^{kl}_{ij}]_{k,l\in\N_m}]_{i,j\in\N_n}$ is a network matrix.
	\end{enumerate}
\end{lemma}

The above lemma enables claiming custom block matrices as ``network matrices'' by enforcing additional conditions. For example, if $M,P$ are block network matrices and $P$ is block diagonal, then: (1) $M^\T P, P M$ and $M^T P + PM$ are network matrices, and (2) if $\scriptsize \Psi\triangleq\bm{P & M^\T P\\ PM & P}$ is some block-block matrix, its \emph{block element-wise} (BEW) form $\mbox{BEW}(\Psi)\triangleq$ $\scriptsize \bm{\bm{P_{ii}e_{ij} & M_{ji}^\T P_{jj} \\ P_{ii}M_{ij} & P_{ii}e_{ij}}}_{i,j\in\N_N}$ is a network matrix.

\subsection{Positive Definiteness}
We next provide several useful lemmas on the positive (or negative) definiteness of a few matrix expressions of interest. 

\begin{lemma}\label{Lm:AlternativeLMI_Schur}
	Let $W \triangleq \Phi^\T \Theta \Phi - \Gamma$ where $\Phi,\Theta$ and $\Gamma$ are arbitrary matrices such that $\Theta > 0$ and $\Gamma = \Gamma^\T$. 
	Then: 
	$$W < 0 \iff 
	\bm{\Theta & \Theta \Phi \\ \Phi^\T \Theta & \Gamma} > 0.$$
\end{lemma}
\begin{proof}	
Using the Schur's complement theory and the congruence principle \cite{Bernstein2009}, we can directly obtain:\\
$
	W<0 \iff \bm{\Theta^{-1} & \Phi \\ \Phi^\T & \Gamma} > 0 \iff 
	\bm{\Theta & \Theta \Phi \\ \Phi^\T \Theta & \Gamma} > 0.
$ 
\end{proof}

\begin{lemma}\label{Lm:AlternativeLMI_LowerBound}
	Let $W \triangleq \Phi^\T \Theta \Phi - \Gamma$ where $\Phi,\Theta$ and $\Gamma$ are arbitrary matrices such that $\Theta < 0$ and $\Gamma = \Gamma^\T$. 
	Then: 
	$$W < 0 \impliedby 
	- \Phi^\T \Theta \Psi - \Psi^\T \Theta \Phi + \Psi^\T \Theta \Psi + \Gamma > 0,$$
	for any matrix $\Psi$ with appropriate dimensions.
\end{lemma}
\begin{proof}
Note that $\Theta<0 \implies (\Phi-\Psi)^\T \Theta (\Phi-\Psi)<0 
\iff \Phi^\T \Theta \Phi - \Gamma < \Phi^\T \Theta \Psi + \Psi^\T \Theta \Phi - \Psi^\T \Theta \Psi -\Gamma$.
Thus, $W<0 \impliedby -\Phi^\T \Theta \Psi - \Psi^\T \Theta \Phi + \Psi^\T \Theta \Psi + \Gamma>0$. 
\end{proof}

\begin{remark}
	 Note that, the use of $\Psi=\0$ in the above lemma yields $W <0 \impliedby \Gamma > 0$. However, the use of $\Psi=\alpha \I$ (assuming $\Phi$ is a square matrix) yields 
	\begin{equation}
		W<0 \impliedby -\alpha\left( \Phi^\T \Theta + \Theta \Phi \right) + \alpha^2 \Theta + \Gamma  > 0.
	\end{equation}
	Hence the latter choice for $\Psi$ (i.e., $\Psi = \alpha \I$) is more effective whenever 
	$-\alpha\left( \Phi^\T \Theta + \Theta \Phi \right) + \alpha^2 \Theta > 0$. Since $\alpha^2 \Theta < 0$, this condition holds whenever $\Phi^\T \Theta + \Theta \Phi>0$ with $\alpha < 0$ or $\Phi^\T \Theta + \Theta \Phi<0$ with $\alpha>0$.  
\end{remark}

\begin{lemma}\label{Lm:AlternativeLMI_BEW}
	\cite{WelikalaP32022}
	Let $\Psi = [\Psi^{kl}]_{k,l\in\N_m}$ be an $m\times m$ block-block matrix where $\Psi^{kl}, \forall  k,l\in\N_m$ are $n \times n$ block matrices. Then, $\Psi > 0 \iff \text{BEW}(\Psi) \triangleq [[\Psi^{kl}_{ij}]_{k,l\in\N_m}]_{i,j\in\N_n}>0$.
\end{lemma}

\begin{remark}
Lemmas \ref{Lm:AlternativeLMI_Schur} and \ref{Lm:AlternativeLMI_LowerBound} can be used to replace a particular form of a tri-linear matrix inequality condition with a bi-linear matrix inequality condition. This will be useful later on when deriving LMI conditions. Moreover, if $\Phi,\Theta,\Psi$ and $\Gamma$ used in Lemmas \ref{Lm:AlternativeLMI_Schur} and \ref{Lm:AlternativeLMI_LowerBound} are block network matrices of some network and $\Phi$ and $\Psi$ are also block diagonal, then, with the help of Lemmas \ref{Lm:NetworkMatrixProperties} and \ref{Lm:AlternativeLMI_BEW}, it is easy to see that the aforementioned bi-linear matrix inequality conditions can also be replaced with a positive-definiteness condition of a network matrix. This will be useful later on when decentralizing different LMIs of interest.
\end{remark}

Inspired by the Sylvester's criterion \cite{Antsaklis2006}, the following lemma provides a decentralized and compositional testing criterion to evaluate the positive definiteness of an $N\times N$ block matrix $W=[W_{ij}]_{i,j\in\N_N}$ (for more details, see  \cite{Welikala2022Ax2}).  

\begin{lemma}\label{Lm:MainLemma}
	\cite{WelikalaP32022}
	A symmetric $N \times N$ block matrix $W = [W_{ij}]_{i,j\in\N_N} > 0$ if and only if
	\begin{equation}\label{Eq:Lm:MainLemma1}
		\tilde{W}_{ii} \triangleq W_{ii} - \tilde{W}_i \mathcal{D}_i \tilde{W}_i^\T > 0,\ \ \ \ \forall i\in\N_N,
	\end{equation}
	where
	\begin{equation}\label{Eq:Lm:MainLemma2}
		\begin{aligned}
			\tilde{W}_i \triangleq&\ [\tilde{W}_{ij}]_{j\in\N_{i-1}} \triangleq W_i(\mathcal{D}_i\mathcal{A}_i^\T)^{-1},\\
			W_i \triangleq&\  [W_{ij}]_{j\in\N_{i-1}}, \ \ \ 
			\mathcal{D}_i \triangleq \diag(\tilde{W}_{jj}^{-1}:j\in\N_{i-1}),\\
			\mathcal{A}_i \triangleq&\ 
			\bm{
				\tilde{W}_{11} & \0 & \cdots & \0 \\
				\tilde{W}_{21} & \tilde{W}_{22} & \cdots & \0\\
				\vdots & \vdots & \vdots & \vdots \\
				\tilde{W}_{i-1,1} & \tilde{W}_{i-1,2} & \cdots & \tilde{W}_{i-1,i-1}
			}.
		\end{aligned}
	\end{equation}
\end{lemma}

\begin{remark}\label{Rm:Lm:MainLemma}
If $W$ is a \emph{symmetric block network matrix} corresponding to some network $\mathcal{G}_N$, the above lemma can be used test $W>0$ in a decentralized and compositional manner (i.e., by sequentially testing $\tilde{W}_{ii}>0$ at each subsystem $\Sigma_i, i\in\N_N$).
\end{remark}

\begin{algorithm}[!t]
	\caption{Testing/Enforcing $W>0$ in a Network Setting.}
	\label{Alg:DistributedPositiveDefiniteness}
	\begin{algorithmic}[1]
		\State \textbf{Input: } $W = [W_{ij}]_{i,j\in\N_N}$
		\State \textbf{At each subsystem $\Sigma_i, i \in \N_N$ execute:} 
		\Indent
		\If{$i=1$}
		\State Test/Enforce: $W_{11}>0$
		\State Store: $\tilde{W}_1 \triangleq [W_{11}]$ \Comment{To be sent to others.}
		\Else
		\State \textbf{From each subsystem $\Sigma_j, j\in\N_{i-1}$:}
		\Indent
		\State Receive: $\tilde{W}_j \triangleq  [\tilde{W}_{j1},\tilde{W}_{j2},\ldots,\tilde{W}_{jj}]$
		\State Receive: Required info. to compute $W_{ij}$
		\EndIndent
		\State \textbf{End receiving}
		\State Construct: $\mathcal{A}_i, \mathcal{D}_i$ and $W_i$.  
		\Comment{Using: \eqref{Eq:Lm:MainLemma2}.}
		\State Compute: $\tilde{W}_i \triangleq W_i (\mathcal{D}_i\mathcal{A}_i^\T)^{-1}$ 
		\Comment{From \eqref{Eq:Lm:MainLemma2}.}
		\State Compute: $\tilde{W}_{ii} \triangleq W_{ii} - \tilde{W}_{i}\mathcal{D}_i\tilde{W}_i^\T$ 
		\Comment{From \eqref{Eq:Lm:MainLemma1}.}
		\State Test/Enforce: $\tilde{W}_{ii} > 0$
		\State Store: $\tilde{W}_i \triangleq [\tilde{W}_i, \tilde{W}_{ii}]$
		\Comment{To be sent to others}
		\EndIf
		\EndIndent
		\State \textbf{End execution}
	\end{algorithmic}
\end{algorithm}

The above lemma shows that testing positive definiteness of an $N\times N$ block matrix $W=[W_{ij}]_{i,j\in\N_N}$ can be broken down to $N$ separate smaller sequence of tests (iterations). 

In a network setting (i.e., where $W$ is a block network matrix corresponding to a network $\mathcal{G}_N$), at the $i$\tsup{th} iteration (i.e., at the subsystem $\Sigma_i$), we now only need to test whether $\tilde{W}_{ii}>0$, where $\tilde{W}_{ii}$ can be computed using: 
(1) $\{W_{ij}:j \in \N_i\}$ (related blocks to the subsystem $\Sigma_i$ extracted from $W$), 
(2) $\{\tilde{W}_{ij}:j\in \N_{i-1}\}$ (computed using \eqref{Eq:Lm:MainLemma2} at subsystem $\Sigma_i$), and 
(3) $\{\{\tilde{W}_{jk}:k\in\N_j\}:j\in\N_{i-1}\}$ (blocks computed using \eqref{Eq:Lm:MainLemma2} at previous subsystems $\{\Sigma_j:j\in\N_{i-1}\}$). 

Therefore, it is clear that Lm. \ref{Lm:MainLemma} can be used to test/enforce the positive definiteness of a network matrix in a decentralized manner. Further, as shown in \cite{Welikala2022Ax2}, for some network topologies, this process is fully distributed (i.e., no communications are required between non-neighboring subsystems). Furthermore, it is established in \cite{WelikalaP32022} that this process is compositional (i.e., resilient to subsystem removals/additions from/to the network). This decentralized and compositional approach to test/enforce the positive-definiteness of a network matrix $W$ is summarized in the Alg. \ref{Alg:DistributedPositiveDefiniteness}.

\section{Networked Systems}\label{Sec:NetworkedSystems}

\subsection{Networked System Configurations}

\begin{figure}[!t]
    \centering
    \begin{subfigure}[t]{0.23\textwidth}
        \centering
        \captionsetup{justification=centering}
        \includegraphics[width=1.0in]{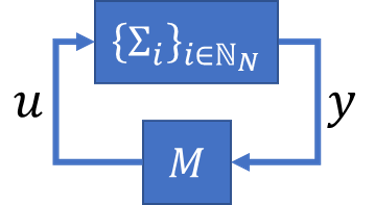}
        \caption{NSC 1}
        \label{Fig:Interconnection1}
    \end{subfigure}
    \hfill
    \begin{subfigure}[t]{0.23\textwidth}
        \centering
        \captionsetup{justification=centering}
        \includegraphics[width=1.4in]{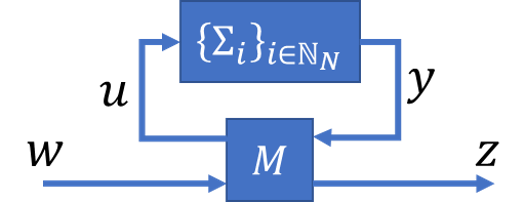}
        \caption{NSC 2}
        \label{Fig:Interconnection2}
    \end{subfigure}
    \begin{subfigure}[t]{0.23\textwidth}
        \centering
        \captionsetup{justification=centering}
        \includegraphics[width=1.0in]{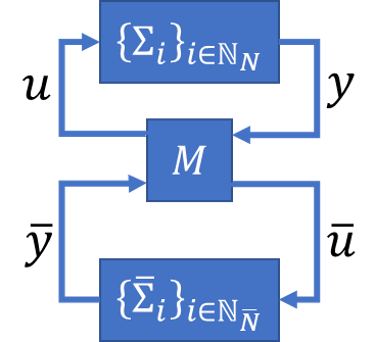}
        \caption{NSC 3}
        \label{Fig:Interconnection3}
    \end{subfigure}
    \hfill
    \begin{subfigure}[t]{0.23\textwidth}
        \centering
        \captionsetup{justification=centering}
        \includegraphics[width=1.4in]{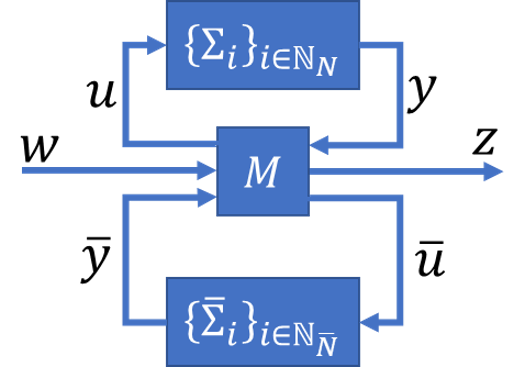}
        \caption{NSC 4}
        \label{Fig:Interconnection4}
    \end{subfigure}
    \caption{Considered four networked (interconnected) system configurations (abbreviated as NSCs 1-4).}
    \label{Fig:InterconnectedSystems}
\end{figure}

In this paper, we consider the four networked (interconnected) system configurations (NSCs) shown in Fig. \ref{Fig:InterconnectedSystems}. Each of these NSC is comprised of independent dynamic subsystems $\Sigma_i,i\in\N_N, \bar{\Sigma}_i, i\in\N_{\bar{N}}$ and a static $M$ matrix that defines how the subsystems, an exogenous input signal  $w(t) \in \R^{r}$ (e.g., disturbance) and interested output signal $z(t) \in\R^{l}$ (e.g., performance) are interconnected, over $t \geq 0$. 

Analogous to \eqref{Eq:GeneralSystem}, dynamics of subsystem $\Sigma_i, i\in\N_N$ are 
\begin{equation}
    \begin{aligned}
        \dot{x}_i(t) = f_i(x_i(t),u_i(t)),\\
        y_i(t) = h_i(x_i(t),u_i(t)),
    \end{aligned}
\end{equation}
where $x_i(t)\in\R^{n_i}, u_i(t)\in\R^{q_i}, y_i(t)\in\R^{m_i}$ and $u \triangleq \bm{u_i^\T}_{i\in\N_N}^\T, y \triangleq \bm{y_i^\T}_{i\in\N_N}^\T$. Further, each subsystem $\Sigma_i, i\in\N_N$ is assumed to have a set $\mathcal{X}_i \subset \R^{n_i}$ where for every $\tilde{x}_i \in \mathcal{X}_i$, there is a unique $\tilde{u}_i \in \R^{q_i}$ satisfying $f_i(\tilde{x}_i,\tilde{u}_i)=0$, and both $\tilde{u}_i$ and $\tilde{y}_i\triangleq h_i(\tilde{x}_i,\tilde{u}_i)$ are implicit functions of $\tilde{x}_i$. Furthermore, subsystem $\Sigma_i, i\in\N_N$ is assumed to be $X_i$-EID, i.e., there is a storage function $V_i:\R^{n_i}\times \X_i\rightarrow \R$ satisfying $\forall(x_i,\tilde{x}_i,u_i)\in \R^{n_i}\times \X_i \times \R^{q_i}:$ $V_i(x_i,\tilde{x}_i)>0$, $V_i(\tilde{x}_i,\tilde{x}_i)=0$, and 
\begin{equation*}
	\dot{V}_i(x_i,\tilde{x}_i) 
	\leq  \bm{u_i-\tilde{u}_i\\y_i-\tilde{y}_i}^\T 
	\underbrace{\bm{X_i^{11} & X_i^{12}\\ X_i^{21} & X_i^{22}}}_{X_i} \bm{u_i-\tilde{u}_i\\y_i-\tilde{y}_i}.
\end{equation*}    
Subsystems $\bar{\Sigma}_i,i\in\N_{\bar{N}}$ are defined similarly but involves quantities denoted with a bar symbol (e.g., subsystem $\bar{\Sigma}_i,i\in\N_{\bar{N}}$ is $\bar{X}_i$-EID with a storage function $\bar{V}_i:\R^{\bar{n}_i} \times \bar{\X}_i \rightarrow \R$). 
	
On the other hand, the format (block structure) of the static $M$ matrix involved in different NSC differs according to the quantities that it interconnects. In particular, for the considered four NSCs shown in Fig. \ref{Fig:InterconnectedSystems}, the interconnection relationships (that involves $M$) are respectively as follows: 

\begin{align}
	\label{Eq:NSC1Interconnection}\mbox{NSC 1}:&\ \ \ u \ \ = M y \equiv M_{uy} y,\\
	\label{Eq:NSC2Interconnection}\mbox{NSC 2}:&\ \bm{u\\z} \, = M \bm{y\\w} \equiv 
	\bm{M_{uy} & M_{uw} \\ M_{zy} & M_{zw}}\bm{y\\w},\\
	\label{Eq:NSC3Interconnection}\mbox{NSC 3}:&\ \bm{u \\ \bar{u}} \ = M \bm{y\\ \bar{y}} \equiv 
	\bm{M_{uy} & M_{u\bar{y}} \\ M_{\bar{u}y} & M_{\bar{u}\bar{y}}}\bm{y\\ \bar{y}},\\
	\label{Eq:NSC4Interconnection}\mbox{NSC 4}:&\ \bm{u \\ \bar{u} \\ z} \, = M \bm{y\\ \bar{y} \\ w} \equiv
	\bm{M_{uy} & M_{u\bar{y}} & M_{uw} \\ 
		M_{\bar{u}y} & M_{\bar{u}\bar{y}} & M_{\bar{u}w}\\
		M_{zy} & M_{z\bar{y}} & M_{zw}}
	\bm{y\\ \bar{y} \\ w}.
\end{align}

\subsection{Connection to Related Work on NSCs}

The work in \cite{Arcak2016,Arcak2022} (and references therein) have already studied the NSCs 1-2 inspired by their applicability in different application domains such as in resource allocation in communications networks, multi-agent motion coordination and biochemical reaction networks. In particular, \cite{Arcak2016,Arcak2022} provide linear matrix inequality (LMI) conditions to analyze the EI-stability of NSC 1 and $X$-EID of NSC 2. However, these LMIs are only applicable for analysis (i.e., when the interconnection matrix $M$ is given) but not for synthesis (i.e., when $M$ needs to be designed), of NSCs. Moreover, even though \cite{Arcak2016,Arcak2022} derive the said LMI conditions in a compositional manner, they can only be evaluated in a centralized setting. These limitations of \cite{Arcak2016,Arcak2022} regarding NSCs 1-2 motivated this paper to focus on these synthesis and decentralization aspects. 

Compared to \cite{Arcak2016,Arcak2022}, the earlier work in \cite{Xia2014,Zakeri2022} has actually considered the problem of synthesizing the interconnection matrix $M$. However, the proposed synthesizing strategy in \cite{Xia2014,Zakeri2022} is limited to the NSC 2 with a single subsystem (i.e., $N=1$). Moreover, this strategy \cite[Th. 3]{Xia2014} involves non-linear inequality conditions, and thus, it is not straightforward. In this paper, we will address these limitations while generalizing the result in \cite[Th. 3]{Xia2014}.

Another interesting contribution of \cite{Xia2014} is that it not only considers the problem of synthesizing $M$ for the NSC 2 (albeit with $N=1$), it also provide conditions under which a synthesized NSC 2 can be used as its controller for another system. This motivated the study of NSCs 3-4 as they can be seen as instants where a NSC 2 has been used as a controller to govern another networked system (comprised of  subsystems $\bar{\Sigma}_i,i\in\bar{B}$). In this paper, we will generalize the result in \cite[Th. 4]{Xia2014} by providing LMI conditions for EI-stability of NSC 3 and $X$-EID of NSC 4, in both centralized and decentralized settings. Moreover, as shown in Tab. \ref{Tab:Configurations}, the NSC 4 (similarly, the NSC 3) can be used to represent several standard control system configurations.

Finally, we also point out that our recent work in \cite{WelikalaP32022} has addressed the decentralization of several standard LMI based control solutions associated with networked systems comprised of linear subsystems. Therefore, this paper can also be viewed as a extension of \cite{WelikalaP32022} for networked systems comprised of non-linear systems.

\begin{table}[!h]
    \centering
    \resizebox{\linewidth}{!}{
    \begin{tabular}{|c|c|}\hline
    Control System Configuration &  \begin{tabular}{@{}c@{}} Format of $M$ \eqref{Eq:NSC4Interconnection} 
    \\ in the NSC 4 
    \end{tabular}\\ \hline
    \begin{tabular}{@{}c@{}} \raisebox{-0.5\totalheight}{\includegraphics[width=2.6in]{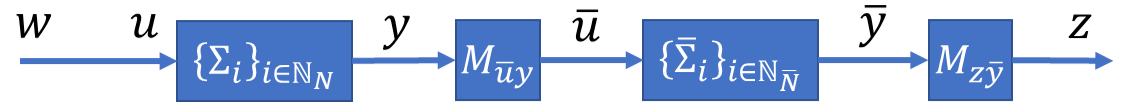}} 
    \\ (Series Configuration) 
    \end{tabular} 
    & $\scriptsize \bm{\0 & -M_{\tilde{u}\bar{y}} & \I \\ M_{\bar{u}y} & \0 & \0 \\ \0 & M_{z\bar{y}} & \0}$ \\ \hline
    \begin{tabular}{@{}c@{}}
    \raisebox{-0.5\totalheight}{\includegraphics[width=2.2in]{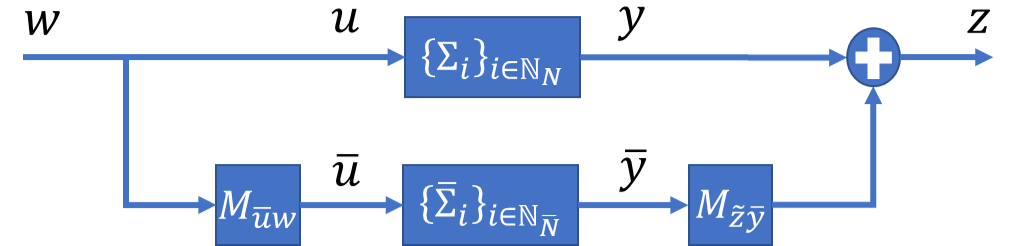}}
    \\ (Parallel Configuration) 
    \end{tabular} 
    & $\scriptsize \bm{\0 & \0 & \I \\ \0 & \0 & M_{\bar{u}w} \\ \I & M_{\tilde{z}y} & \0}$ \\ \hline
    \begin{tabular}{@{}c@{}}
    \raisebox{-0.5\totalheight}{\includegraphics[width=2.1in]{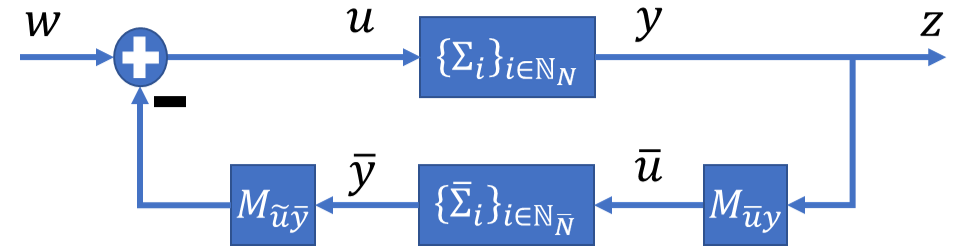}}
    \\ (Feedback Configuration) 
    \end{tabular} 
    & $\scriptsize \bm{\0 & -M_{\tilde{u}\bar{y}} & \I \\ M_{\bar{u}y} & \0 & \0 \\ \I & \0 & \0}$ \\ \hline
    \begin{tabular}{@{}c@{}}
    \raisebox{-0.5\totalheight}{\includegraphics[width=1.5in]{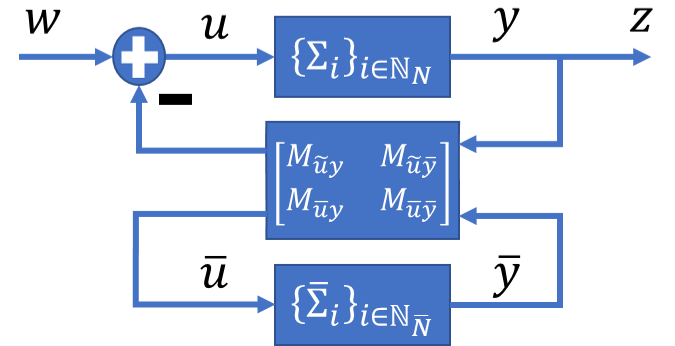}}
    \\ (Feedback Reconfiguration Configuration \cite{Xia2014}) 
    \end{tabular} 
    & $\scriptsize \bm{-M_{\tilde{u}y} & -M_{\tilde{u}\bar{y}} & \I \\ M_{\bar{u}y} & M_{\bar{u}\bar{y}} & \0 \\ \I & \0 & \0}$ \\ \hline
    \begin{tabular}{@{}c@{}}
    \raisebox{-0.5\totalheight}{\includegraphics[width=1.8in]{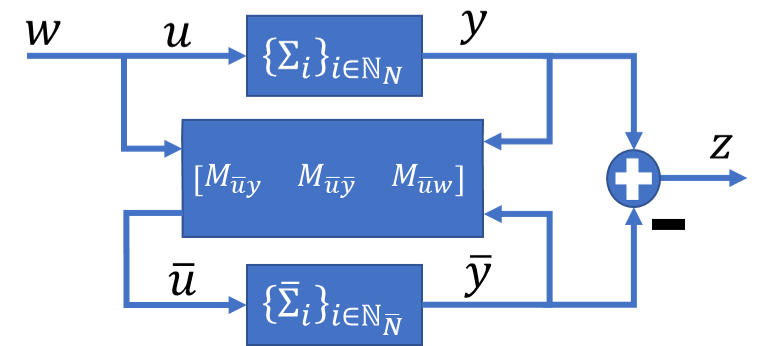}}
    \\ (Approximate Simulation Configuration \cite{Song2022}) 
    \end{tabular} 
    & $\scriptsize \bm{\0 & \0 & \I \\ M_{\bar{u}y} & M_{\bar{u}\bar{y}} & M_{\bar{u}w} \\ \I & -\I & \0}$ \\ \hline
    \end{tabular}}
    \caption{Some standard control system configurations and the corresponding format of $M$ \eqref{Eq:NSC4Interconnection} in the NSC 4.}
    \label{Tab:Configurations}
\end{table}

\subsection{Centralized Stability and Dissipativity Analysis of NSCs}

The following four propositions provide centralized LMI conditions for stability and $\textbf{Y}$-dissipativity analysis of NSCs 1-4 (using the EID properties of the subsystems). It is noteworthy that these analysis techniques are respectively independent of the equilibrium points of the NSCs 1-4.

 
\begin{proposition}\label{Pr:NSC1Stability}\hspace{-2mm}\cite{Arcak2006}
Let $x^* \triangleq [x_i^{*\T}]_{i\in\N_N}^\T$ be an equilibrium point of NSC 1 with $x_i^* \in\mathcal{X}_i, \forall i\in\N_N$.  Then, $x^*$ is stable for NSC 1 if there exist scalars $p_i>0, \forall i \in \N_N$ such that 
\begin{equation}\label{Eq:Pr:NSC1Stability1}
    \bm{M_{uy}\\ \I}^\T \textbf{X}_p \bm{M_{uy}\\ \I} \leq 0,
\end{equation}
where components of $M$ are from \eqref{Eq:NSC3Interconnection} and 
\begin{equation}\label{Eq:Pr:NSC1Stability2}
    \textbf{X}_p \triangleq \scriptsize
    \bm{
    \diag(p_i X_i^{11}:i\in\N_N) & \diag(p_i X_i^{12}:i\in\N_N)\\ 
    \diag(p_i X_i^{21}:i\in\N_N) & \diag(p_i X_i^{22}:i\in\N_N)\\ 
    }.\normalsize
\end{equation}
\end{proposition}


\begin{proposition}\label{Pr:NSC2Dissipativity}\cite{Arcak2016}
Let $x^* \triangleq [x_i^{*\T}]_{i\in\N_N}^\T$ be an equilibrium point of NSC 2 (under $w(t)=\0$) with $x_i^* \in\mathcal{X}_i, \forall i\in\N_N$ and $z^*$ be the corresponding output. Then, the NSC 2 is $\textbf{Y}$-dissipative, i.e., dissipative with respect to the supply rate:
\begin{equation}\label{Eq:Pr:NSC2Dissipativity0}
    \bm{w\\z-z^*}^\T \textbf{Y} \bm{w\\z-z^*}, 
\end{equation}
if there exist scalars $p_i \geq 0, \forall i \in \N_N$ such that 
\begin{equation}\label{Eq:Pr:NSC2Dissipativity1}
\bm{M_{uy} & M_{uw} \\ \I & \0 \\ 
\0 & \I \\ M_{zy} & M_{zw}}^\T
\bm{\textbf{X}_p & \0 \\ \0 & -\textbf{Y}}
\bm{M_{uy} & M_{uw} \\ \I & \0 \\ 
\0 & \I \\ M_{zy} & M_{zw}} \leq 0,
\end{equation}
where components of $M$ are from \eqref{Eq:NSC2Interconnection} and $\textbf{X}_p$ is as in \eqref{Eq:Pr:NSC1Stability2}.
\end{proposition}


\begin{proposition}\label{Pr:NSC3Stability}
Let $\hat{x}^* \triangleq \left[[x_i^{*\T}]_{i\in\N_N}, \  [\bar{x}_i^{*\T}]_{i\in\N_{\bar{N}}}\right]^\T$ be an equilibrium point of NSC 3 with $x_i^* \in \mathcal{X}_i, \forall i\in\N_N$ and $\bar{x}_i^* \in \bar{\mathcal{X}}_i, \forall i\in\N_{\bar{N}}$. Then, $\hat{x}^*$ is stable for the NSC 3 if there exist scalars $p_i > 0, \forall i \in \N_N$ and $\bar{p}_i > 0, \forall i \in \N_{\bar{N}}$ such that
\begin{equation}\label{Eq:Pr:NSC3Stability1}
    \bm{M_{uy} & M_{u\bar{y}} \\ \I & \0 \\ M_{\bar{u}y} & M_{\bar{u}\bar{y}} \\ \0 & \I}^\T 
    \bm{\textbf{X}_p & \0 \\ \0 & \bar{\textbf{X}}_{\bar{p}}} 
    \bm{M_{uy} & M_{u\bar{y}} \\ \I & \0 \\ M_{\bar{u}y} & M_{\bar{u}\bar{y}} \\ \0 & \I} \leq 0,
\end{equation}
where components of $M$ are from \eqref{Eq:NSC3Interconnection}, $\textbf{X}_p$ is as in \eqref{Eq:Pr:NSC1Stability2} and 
\begin{equation}\label{Eq:Pr:NSC3Stability2}
    \bar{\textbf{X}}_{\bar{p}} \triangleq  \scriptsize
    \bm{
    \diag(\bar{p}_i\bar{X}_i^{11}:i\in\N_N) & \diag(\bar{p}_i\bar{X}_i^{12}:i\in\N_N)\\ 
    \diag(\bar{p}_i\bar{X}_i^{21}:i\in\N_N) & \diag(\bar{p}_i\bar{X}_i^{22}:i\in\N_N)
    }.\normalsize
\end{equation}
\end{proposition}
\begin{proof}
Consider the Lyapunov function 
\begin{equation}\label{Eq:Pr:NSC3StabilityStep1}
    V(\hat{x}) \triangleq \sum_{i\in\N_N} p_i V_i(x_i,x_i^*) + \sum_{i\in\N_{\bar{N}}} \bar{p}_i \bar{V}_i(\bar{x}_i,\bar{x}_i^*)
\end{equation}
constructed from the storage functions of subsystems $\Sigma_i, i\in\N_N$ and $\bar{\Sigma}_i, i\in\N_{\bar{N}}$. Due to the EID of subsystems, we get: 
\begin{align}
    &\dot{V}(\hat{x}) \leq 
    \sum_{i\in\N_N} p_i \bm{u_i-u_i^*\\y_i-y_i^*}^\T X_i \bm{u_i-u_i^*\\y_i-y_i^*}
    \nonumber \\
    &+ \sum_{i\in\N_{\bar{N}}} \bar{p}_i \bm{\bar{u}_i-\bar{u}_i^*\\\bar{y}_i-\bar{y}_i^*}^\T \bar{X}_i \bm{\bar{u}_i-\bar{u}_i^*\\\bar{y}_i-\bar{y}_i^*} \nonumber\\
    &= \bm{u-u^*\\y-y^*}^\T \hspace{-1mm}\textbf{X}_p\hspace{-1mm} \bm{u-u^*\\y-y^*}
    + \bm{\bar{u}-\bar{u}^*\\\bar{y}-\bar{y}^*}^\T \hspace{-1mm}\bar{\textbf{X}}_{\bar{p}}\hspace{-1mm} \bm{\bar{u}-\bar{u}^*\\\bar{y}-\bar{y}^*}. \label{Eq:Pr:NSC3StabilityStep2}
\end{align}
Note that the quantities with a superscript ``$*$'' correspond to the assumed equilibrium point $\hat{x}^*$. Using \eqref{Eq:NSC3Interconnection}, we can write 
\begin{equation}\label{Eq:Pr:NSC3StabilityStep3}
\begin{aligned}
    \bm{u-u^*\\y-y^*} = 
    \bm{M_{uy} & M_{u\bar{y}} \\ \I & \0}
    \bm{y-y^*\\ \bar{y}-\bar{y}^*},\\
    \bm{\bar{u}-\bar{u}^*\\\bar{y}-\bar{y}^*} = 
    \bm{M_{\bar{u}y} & M_{\bar{u}\bar{y}} \\ \0 & \I} 
    \bm{y-y^*\\ \bar{y}-\bar{y}^*}.
\end{aligned}
\end{equation}
Applying \eqref{Eq:Pr:NSC3StabilityStep3} in \eqref{Eq:Pr:NSC3StabilityStep2} and combining the two terms, we get: 
\begin{equation*}
\begin{aligned}
    \dot{V}(\hat{x}) \leq   
    \bm{y-y^*\\ \bar{y}-\bar{y}^*}^\T 
    \bm{M_{uy} & M_{u\bar{y}} \\ \I & \0 \\ M_{\bar{u}y} & M_{\bar{u}\bar{y}} \\ \0 & \I}^\T 
    \bm{\textbf{X}_p & \0 \\ \0 & \bar{\textbf{X}}_{\bar{p}}} \, \star.
\end{aligned}
\end{equation*}
Thus, \eqref{Eq:Pr:NSC3Stability1} implies $\dot{V}(\hat{x}) \leq 0$. Further, by definition \eqref{Eq:Pr:NSC3StabilityStep1}, $\dot{V}(\hat{x}^*) = 0$. Hence $\hat{x}^*$ for the NSC 3 is Lyapunov stable.  
\end{proof}


\begin{proposition}\label{Pr:NSC4Dissipativity}
Let $\hat{x}^* \triangleq \left[[x_i^{*\T}]_{i\in\N_N}, [\bar{x}_i^{*\T}]_{i\in\N_{\bar{N}}}\right]^\T$ be an equilibrium point of NSC 4 (under $w(t)=0$) with $x_i^* \in \mathcal{X}_i, \forall i\in\N_N$, $\bar{x}_i^* \in \bar{\mathcal{X}}_i, \forall i\in\N_{\bar{N}}$ and $z^*$ be the corresponding output. Then, the NSC 4 is $\textbf{Y}$-dissipative, i.e., dissipative with respect to the supply rate \eqref{Eq:Pr:NSC2Dissipativity0}, 
if there exist scalars $p_i \geq 0, \forall i \in \N_N$ and $\bar{p}_i, \forall i\in\N_{\bar{N}}$ such that 
\begin{equation}\label{Eq:Pr:NSC4Dissipativity}
    \bm{
    M_{uy} & M_{u\bar{y}} & M_{uw} \\ \I & \0 & \0 \\
    M_{\bar{u}y} & M_{\bar{u}\bar{y}} & M_{\bar{u}w} \\ \0 & \I & \0 \\
    \0 & \0 & \I\\ M_{zy} & M_{z\bar{y}} & M_{zw} 
    }^\T
    \bm{
    \textbf{X}_p & \0 & \0 \\ \0 & \bar{\textbf{X}}_{\bar{p}} & \0 \\ \0 & \0 & -\textbf{Y}
    }
    \, \star \leq 0,
\end{equation}
where components of $M$ are from \eqref{Eq:NSC4Interconnection}, and $\textbf{X}_p$ and $\bar{\textbf{X}}_{\bar{p}}$ are as defined in \eqref{Eq:Pr:NSC1Stability2} and \eqref{Eq:Pr:NSC3Stability2}, respectively.
\end{proposition}

\begin{proof}
Consider $V(\hat{x})$ in \eqref{Eq:Pr:NSC3StabilityStep1} now as a storage function. Using the EID of subsystems, we can obtain an upper-bound for $\dot{V}(\hat{x})$ as in \eqref{Eq:Pr:NSC3StabilityStep2}. Thus, the NSC 4 is $\textbf{Y}$-dissipative if  
\begin{equation}\label{Pr:NSC4DissipativityStep1}
    \bm{u-u^*\\y-y^*}^\T \textbf{X}_p \,\star 
    + \bm{\bar{u}-\bar{u}^*\\\bar{y}-\bar{y}^*}^\T \bar{\textbf{X}}_{\bar{p}} \,\star  
    \leq 
    \bm{w\\z-z^*}^\T \textbf{Y} \,\star.
\end{equation}

On the other hand, using \eqref{Eq:NSC4Interconnection}, we can write  
\begin{equation}\label{Pr:NSC4DissipativityStep2}
\begin{aligned}
    \bm{u-u^*\\y-y^*} = 
    \bm{M_{uy} & M_{u\bar{y}} & M_{uw} \\ \I & \0 & \0}
    \bm{y-y^*\\ \bar{y}-\bar{y}^* \\ w},\\
    \bm{\bar{u}-\bar{u}^*\\\bar{y}-\bar{y}^*} = 
    \bm{M_{\bar{u}y} & M_{\bar{u}\bar{y}} & M_{\bar{u}w} \\ \0 & \I & \0 } 
    \bm{y-y^*\\ \bar{y}-\bar{y}^* \\ w},\\
    \bm{w\\z-z^*} = 
    \bm{\0 & \0 & \I \\ M_{zy} & M_{z\bar{y}} & M_{zw}} 
    \bm{y-y^*\\ \bar{y}-\bar{y}^* \\ w}.
\end{aligned}
\end{equation}
Applying \eqref{Pr:NSC4DissipativityStep2} in \eqref{Pr:NSC4DissipativityStep1} and doing some simplifications, we can obtain the condition for NSC 4 to be $\textbf{Y}$-dissipative as: 
\begin{equation*}
    \bm{y-y^* \\ \bar{y}-\bar{y}^* \\ w}^\T 
    \hspace{-1mm}
    \bm{
    M_{uy} & M_{u\bar{y}} & M_{uw} \\ \I & \0 & \0 \\
    M_{\bar{u}y} & M_{\bar{u}\bar{y}} & M_{\bar{u}w} \\ \0 & \I & \0 \\
    \0 & \0 & \I\\ M_{zy} & M_{z\bar{y}} & M_{zw} 
    }^\T 
    \hspace{-1mm}
    \bm{
    \textbf{X}_p & \0 & \0 \\ \0 & \bar{\textbf{X}}_{\bar{p}} & \0 \\ \0 & \0 & -\textbf{Y}
    } \star \leq 0.
\end{equation*}
Therefore, \eqref{Eq:Pr:NSC4Dissipativity} implies that  $\dot{V}(\hat{x})$ is bounded by the storage function \eqref{Eq:Pr:NSC2Dissipativity0}, in other words, NSC 4 is $\textbf{Y}$-dissipative.
\end{proof}

\begin{remark}\label{Rm:StrictNegativeDefiniteness}
    For the analysis of NSCs 1-4, Props. \ref{Pr:NSC1Stability}-\ref{Pr:NSC4Dissipativity} provide LMI conditions of the form $W \leq 0$. In the sequel, we address rather practical problems such as the synthesis of $M$ and the decentralized analysis/synthesis of NSCs 1-4. Thus, henceforth, we consider the LMI conditions given in Props. \ref{Pr:NSC1Stability}-\ref{Pr:NSC4Dissipativity} to be of the form $W<0$ (stricter than before). 
\end{remark}

\section{Centralized Synthesis of Networked Systems} \label{Sec:CentralizedSynthesis}

In this section, we consider the problem of synthesizing the interconnection matrix $M$ \eqref{Eq:NSC1Interconnection}-\eqref{Eq:NSC4Interconnection} associated with the NSCs 1-4 so as to enforce stability or $\textbf{Y}$-dissipativity. To this end, we first introduce some additional notations as follows.

Let us denote the components of: 
$\textbf{X}_p$ in \eqref{Eq:Pr:NSC1Stability2} as $\textbf{X}_p \equiv [\textbf{X}_p^{kl}]_{k,l\in\N_2}$ (i.e., $\textbf{X}_p^{kl} \triangleq \diag(p_i X_i^{kl}:i\in\N_N),\forall k,l \in\N_2$), 
$\bar{\textbf{X}}_{\bar{p}}$ in \eqref{Eq:Pr:NSC3Stability2} as $\bar{\textbf{X}}_{\bar{p}} \equiv [\bar{\textbf{X}}_{\bar{p}}^{kl}]_{k,l\in\N_2}$ (i.e., $\bar{\textbf{X}}_{\bar{p}}^{kl} \triangleq \diag(\bar{p}_i \bar{X}_i^{kl}:i\in\N_N), \forall k,l \in\N_2$), and 
$\textbf{Y}$ in \eqref{Eq:Pr:NSC2Dissipativity0} as $\textbf{Y} \equiv  [\textbf{Y}^{kl}]_{k,l\in\N_2}$. Let us also denote: 
$\textbf{X}^{12} \triangleq \diag((X_i^{11})^{-1}X_i^{12}:i\in\N_N)$, 
$\textbf{X}^{21} \triangleq \diag(X_i^{21}(X_i^{11})^{-1}:i\in\N_N)$, 
$\bar{\textbf{X}}^{12} \triangleq \diag((\bar{X}_i^{11})^{-1}\bar{X}_i^{12}:i\in\N_{\bar{N}})$, and 
$\bar{\textbf{X}}^{21} \triangleq \diag(\bar{X}_i^{21}(\bar{X}_i^{11})^{-1}:i\in\N_{\bar{N}})$.

Further, we require the following assumption regarding the dissipativity properties of the subsystems $\{\Sigma_i: i\in\N_N\}$.   
 
\begin{assumption}\label{As:PositiveDissipativity}
Each subsystem $\Sigma_i, i\in\N_N$ is $X_i$-EID with $X_i^{11}>0$. Also, each subsystem $\bar{\Sigma}_i, i\in\N_{\bar{N}}$ is $\bar{X}_i$-EID with $\bar{X}_i^{11}>0$.
\end{assumption}

\begin{remark}\label{Rm:As:PositiveDissipativity}
According to Rm. \ref{Rm:X-DissipativityVersions}, the above assumption directly holds if each subsystem (say $\Sigma_i, i\in\N_N$) is: (1) L2G($\gamma_i$) (as $X_i^{11} = \gamma_i^2 \I > 0$), or (2) IFP($\nu_i$) with $\nu_i<0$ (i.e., non-passive, as $X_i^{11}=-\nu_i\I > 0$). Therefore, As.  \ref{As:PositiveDissipativity} holds in many cases of interest. 
Note also that, even if a subsystem (say $\Sigma_i, i\in\N_N$) is IFP($\nu_i$) with $\nu_i \geq 0$ (i.e., passive), using the fact that IFP($\nu_i$) $\implies$ IFP($\nu_i-\epsilon_i$) for any $\epsilon_i>0$ \cite[Rm. 1]{Welikala2022Ax1}, As. \ref{As:PositiveDissipativity} can still be upheld by re-assigning the passivity indices. Another less conservative approach to handle this ``passive'' case is given later on in Rm. \ref{Rm:As:PositiveDissipativityFailure}.        
\end{remark}

Finally, we also require the following assumption on $\textbf{Y}$.

\begin{assumption}\label{As:NegativeDissipativity}
The dissipativity specification for the NSCs 2 and 4 (i.e., $\textbf{Y}$-dissipativity) is such that $\textbf{Y}^{22}<0$. 
\end{assumption} 

\begin{remark}\label{Rm:As:NegativeDissipativity}
According to Rm. \ref{Rm:X-DissipativityVersions}, the above assumption holds if we expect the NSCs 2 and 4 to be: (1) L2G($\gamma$) (as $\textbf{Y}^{22} = -\I < 0$), or (2) OFP($\rho$) with $\rho>0$ (i.e., passive, as $\textbf{Y}^{22} = -\rho \I < 0$). Note that it is always desirable to make a system $L_2$-stable or passive. Therefore, As. \ref{As:NegativeDissipativity} is mild.
\end{remark}

\subsection{Centralized Synthesis of NSCs}

The following four propositions provide centralized LMI problems that can be used to synthesize the interconnection matrix $M$ \eqref{Eq:NSC1Interconnection}-\eqref{Eq:NSC4Interconnection}  for the NSCs 1-4, respectively.


\begin{proposition}\label{Pr:NSC1Synthesis}
Under As. \ref{As:PositiveDissipativity}, a stabilizing interconnection matrix $M_{uy}$ \eqref{Eq:NSC1Interconnection} for the NSC 1 can be found via the LMI:
\begin{equation}\label{Eq:Pr:NSC1Synthesis}
\begin{aligned}
    \mbox{Find: }& L_{uy}, \{p_i:  i\in\N_N\}\\
    \mbox{ Sub to: }& p_i > 0, \forall i\in\N_N,\\
    & \bm{ \textbf{X}_p^{11} &  L_{uy} \\ L_{uy}^\T & -(L_{uy}^\T \textbf{X}^{12} + \textbf{X}^{21}L_{uy} + \textbf{X}_p^{22})
	}>0,
\end{aligned}
\end{equation}
as $M_{uy} \triangleq (\textbf{X}_p^{11})^{-1} L_{uy}$.
\end{proposition}

\begin{proof}
Using the new notations, we can restate the stability condition for the NSC 1  \eqref{Eq:Pr:NSC1Stability1} as (see also Rm. \ref{Rm:StrictNegativeDefiniteness}):
\begin{equation}\label{Pr:NSC1SynthesisStep1}
    M_{uy}^\T \textbf{X}_p^{11} M_{uy} + M^\T_{uy} \textbf{X}_p^{12} + \textbf{X}_p^{21} M_{uy} + \textbf{X}_p^{22} < 0.
\end{equation}
Under As. \ref{As:PositiveDissipativity}, $\textbf{X}_p^{11} > 0$, and thus, Lm. \ref{Lm:AlternativeLMI_Schur} can be used to write an equivalent condition to \eqref{Pr:NSC1SynthesisStep1} as: 
\begin{equation}\label{Pr:NSC1SynthesisStep2}
    \bm{ \textbf{X}_p^{11} &  \textbf{X}_p^{11} M_{uy} \\ M_{uy}^\T \textbf{X}_p^{11} & -(M^\T_{uy} \textbf{X}_p^{12} + \textbf{X}_p^{21} M_{uy} + \textbf{X}_p^{22})}>0.
\end{equation}
Since the objective is to design $M_{uy}$, we need the change of variables 
$L_{uy} \triangleq \textbf{X}_p^{11}M_{uy}$. Consequently, $\textbf{X}_p^{21}M_{uy} = \textbf{X}_p^{21}(\textbf{X}_p^{11})^{-1} L_{uy} = \textbf{X}^{21}L_{uy}$  and 
$M^\T_{uy} \textbf{X}_p^{12} = L_{uy}^\T (\textbf{X}_p^{11})^{-1}\textbf{X}_p^{12} = L_{uy}^\T\textbf{X}^{12}$. Thus, \eqref{Pr:NSC1SynthesisStep2} can be formulated as the LMI (in $\{p_i:i\in\N_N\}$ and $L_{uy}$) given in \eqref{Eq:Pr:NSC1Synthesis}.
\end{proof}



\begin{proposition}\label{Pr:NSC2Synthesis}
Under As. \ref{As:PositiveDissipativity} and \ref{As:NegativeDissipativity}, the NSC 2 can be made $\textbf{Y}$-dissipative by synthesizing the interconnection matrix $M \equiv \scriptsize \bm{M_{uy} & M_{uw}\\ M_{zy} & M_{zw}}$  \eqref{Eq:NSC2Interconnection} via the LMI:
\begin{equation}\label{Eq:Pr:NSC2Synthesis}
    \begin{aligned}
    \mbox{Find: }& L_{uy}, L_{uw}, M_{zy}, M_{zw}, \{p_i: i\in\N_N\}\\
    \mbox{ Sub. to: }& p_i > 0, \forall i\in\N_N, \mbox{ and  \eqref{Eq:Pr:NSC2Synthesis2}},
    \end{aligned}
\end{equation}
with $M_{uy} \triangleq (\textbf{X}_p^{11})^{-1} L_{uy}$ and $M_{uw} \triangleq  (\textbf{X}_p^{11})^{-1} L_{uw}$.
\end{proposition}

\begin{proof}
Using the introduced notations, we can restate the condition for $\textbf{Y}$-dissipativity of the NSC 2 (i.e., \eqref{Eq:Pr:NSC2Dissipativity1}) as
\begin{equation}\label{Eq:Pr:NSC2SynthesisStep1}
    \begin{aligned}
    \bm{M_{uy} & M_{uw}\\ M_{zy} & M_{zw}}^\T 
    \bm{\textbf{X}_p^{11} & \0 \\ \0 & -\textbf{Y}^{22}}
    \bm{M_{uy} & M_{uw}\\ M_{zy} & M_{zw}}\\
    + \bm{ M_{uy}^\T \textbf{X}_p^{12} + \textbf{X}_p^{21}M_{uy} + \textbf{X}_p^{22} &  \textbf{X}_p^{21}M_{uw} \\ M_{uw}^\T \textbf{X}_p^{12} & \0}\\
    - \bm{\0  & M_{zy}^\T \textbf{Y}^{21} \\ \textbf{Y}^{12}M_{zy} & M_{zw}^\T \textbf{Y}^{21} +  \textbf{Y}^{12}M_{zw} + \textbf{Y}^{11}}
    < 0.     
    \end{aligned}
\end{equation}
Under As. \ref{As:PositiveDissipativity} and \ref{As:NegativeDissipativity}, $\scriptsize \bm{\textbf{X}_p^{11} & \0 \\ \0 & -\textbf{Y}^{22}}>0$, and thus, Lm. \ref{Lm:AlternativeLMI_Schur} can be used to write an equivalent condition to \eqref{Eq:Pr:NSC2SynthesisStep1}. This equivalent condition, under the change of variables $L_{uy} \triangleq \textbf{X}_p^{11}M_{uy}$ and $L_{uw} \triangleq \textbf{X}_p^{11}M_{uw}$, can be re-formulated as the LMI (in $\{p_i:i\in\N_N\},L_{uy},L_{uw},M_{zy}$ and $M_{zw}$) given in \eqref{Eq:Pr:NSC2Synthesis}.
\end{proof}

\begin{proposition}\label{Pr:NSC3Synthesis}
	Under As. \ref{As:PositiveDissipativity}, a stabilizing interconnection matrix $M \equiv \scriptsize \bm{M_{uy} & M_{u\bar{y}} \\ M_{\bar{u}y} & M_{\bar{u}\bar{y}}}$ \eqref{Eq:NSC3Interconnection} for the NSC 3 can be found via the LMI:
	\begin{equation}\label{Eq:Pr:NSC3Synthesis1}
		\begin{aligned}
			\mbox{Find: } &L_{uy},L_{u\bar{y}},L_{\bar{u}y},L_{\bar{u}\bar{y}},\{p_i:i\in\N_N\},\{\bar{p}_i:i\in\N_{\bar{N}}\}\\
			\mbox{Sub. to: } &p_i>0, \forall i\in\N_N, \ \ \bar{p}_i>0, \forall i\in\N_{\bar{N}}, \mbox{ and } \eqref{Eq:Pr:NSC3Synthesis2},
		\end{aligned}
	\end{equation}
    as
    $
	\scriptsize \bm{M_{uy} & M_{u\bar{y}} \\ M_{\bar{u}y} & M_{\bar{u}\bar{y}}}
	= 
	\bm{\textbf{X}_p^{11} & \0 \\ \0 & \bar{\textbf{X}}_{\bar{p}}^{11}}^{-1}
	\bm{L_{uy} & L_{u\bar{y}} \\ L_{\bar{u}y} & L_{\bar{u}\bar{y}}}
    $.
\end{proposition}

\begin{proof}	
Using the introduced notations, we can re-state the stability condition for the NSC 3 \eqref{Eq:Pr:NSC3Stability1} as  
\begin{equation}\label{Eq:Pr:NSC3SynthesisStep1}
	\begin{aligned}
		\bm{M_{uy} & M_{u\bar{y}} \\ M_{\bar{u}y} & M_{\bar{u}\bar{y}}}^\T 
		\bm{\textbf{X}_p^{11} & \0 \\ \0 & \bar{\textbf{X}}_{\bar{p}}^{11}}
		\bm{M_{uy} & M_{u\bar{y}} \\ M_{\bar{u}y} & M_{\bar{u}\bar{y}}}\\
		+ 
		\bm{M_{uy}^\T \textbf{X}_p^{12}+\textbf{X}_p^{21} M_{uy}+\textbf{X}_p^{22} & \textbf{X}_p^{21} M_{u\bar{y}} \\ M_{u\bar{y}}^\T \textbf{X}_p^{12} & \0} \\
		+
		\bm{\0 & M_{\bar{u}y}^\T \bar{\textbf{X}}_{\bar{p}}^{12} \\ \bar{\textbf{X}}_{\bar{p}}^{21} M_{\bar{u}y} & 
			M_{\bar{u}\bar{y}}^\T \bar{\textbf{X}}_{\bar{p}}^{12} + \bar{\textbf{X}}_{\bar{p}}^{21} M_{\bar{u}\bar{y}} + \bar{\textbf{X}}_{\bar{p}}^{22}} < 0.
	\end{aligned}
\end{equation}
Under As. \ref{As:PositiveDissipativity}, $\scriptsize \bm{\textbf{X}_p^{11} & \0 \\ \0 & \bar{\textbf{X}}_{\bar{p}}^{11}}>0$, and thus, Lm. \ref{Lm:AlternativeLMI_Schur} is applicable here to write an equivalent condition - which, under the change of variables: 
$
	\scriptsize 
	\bm{L_{uy} & L_{u\bar{y}} \\ L_{\bar{u}y} & L_{\bar{u}\bar{y}}} = 
	\bm{\textbf{X}_p^{11} & \0 \\ \0 & \bar{\textbf{X}}_{\bar{p}}^{11}}
	\bm{M_{uy} & M_{u\bar{y}} \\ M_{\bar{u}y} & M_{\bar{u}\bar{y}}}
$, 
can be re-formulated as the LMI given in \eqref{Eq:Pr:NSC3Synthesis1}.	
\end{proof}

\begin{proposition}\label{Pr:NSC4Synthesis}
Under As. \ref{As:PositiveDissipativity} and As. \ref{As:NegativeDissipativity}, the NSC 4 can be made $\textbf{Y}$-dissipative by synthesizing the interconnection matrix $M$ \eqref{Eq:NSC4Interconnection} via the LMI:
\begin{equation}\label{Eq:Pr:NSC4Synthesis1}
\begin{aligned}
	\mbox{Find: } 
	&L_{uy}, L_{u\bar{y}}, L_{uw}, L_{\bar{u}y}, L_{\bar{u}\bar{y}}, L_{\bar{u}w}, M_{zy}, M_{z\bar{y}}, \\
	&M_{zw}, \{p_i:i\in\N_N\}, \{\bar{p}_i:i\in\N_{\bar{N}}\}	\\
	\mbox{Sub. to: } &p_i \geq 0, \forall i\in\N_N, \ \ 
	\bar{p}_i \geq 0, \forall i\in\N_{\bar{N}},\ \  \eqref{Eq:Pr:NSC4Synthesis2},
\end{aligned}
\end{equation}
with 
$
\scriptsize
\bm{M_{uy} & M_{u\bar{y}} & M_{uw} \\ M_{\bar{u}y} & M_{\bar{u}\bar{y}} & M_{\bar{u}w}} = 
\bm{\textbf{X}_p^{11} & \0 \\ \0 & \bar{\textbf{X}}_{\bar{p}}^{11}}^{-1} \hspace{-1mm} \bm{L_{uy} & L_{u\bar{y}} & L_{uw} \\ L_{\bar{u}y} & L_{\bar{u}\bar{y}} & L_{\bar{u}w}}
$.
\end{proposition}

\begin{proof}
Using the introduced notations, we can re-state the condition for $\textbf{Y}$-dissipativity of the NSC 4 (i.e., \eqref{Eq:Pr:NSC4Dissipativity}) as  
\begin{equation}\label{Eq:Pr:NSC4SynthesisStep1}
	\begin{aligned}
		\bm{M_{uy} & M_{u\bar{y}} & M_{uw} \\ M_{\bar{u}y} & M_{\bar{u}\bar{y}} & M_{\bar{u}w} \\ M_{zy} & M_{z\bar{y}} & M_{zw}}^\T	
		\bm{\textbf{X}_p^{11} & \0 & \0 \\ \0 & \bar{\textbf{X}}_{\bar{p}}^{11} & \0 \\ \0 & \0 & -\textbf{Y}^{22}}\star \\
		+\bm{M_{uy}^\T \textbf{X}_p^{12}+\textbf{X}_p^{21}M_{uy}+\textbf{X}_p^{22} & \textbf{X}_p^{21}M_{u\bar{y}} & \textbf{X}_p^{21}M_{uw}\\ M_{u\bar{y}}^\T \textbf{X}_p^{12} & \0 & \0 \\ M_{uw}^\T \textbf{X}_p^{12} & \0 & \0}\\
		+\bm{\0 & M_{\bar{u}y}^\T \bar{\textbf{X}}_{\bar{p}}^{12} & \0 \\ \bar{\textbf{X}}_{\bar{p}}^{21}M_{\bar{u}y} & M_{\bar{u}\bar{y}}^\T \bar{\textbf{X}}_{\bar{p}}^{12} + \bar{\textbf{X}}_{\bar{p}}^{12}M_{\bar{u}\bar{y}}+\bar{\textbf{X}}_{\bar{p}}^{22} & \bar{\textbf{X}}_{\bar{p}}^{21} M_{\bar{u}w}\\ \0 & M_{\bar{u}w}^\T\bar{\textbf{X}}_{\bar{p}}^{12} & \0}\\
		-\bm{\0 & \0 & M_{zy}^\T \textbf{Y}^{21}\\ \0 & \0 & M_{z\bar{y}}^\T \textbf{Y}^{21}\\ \textbf{Y}^{12}M_{zy} & \textbf{Y}^{12}M_{z\bar{y}} & M_{zw}^\T\textbf{Y}^{21} + \textbf{Y}^{12}M_{zw} + \textbf{Y}^{11} } < 0.
	\end{aligned}
\end{equation}
Similar to the proof of Prop. \ref{Pr:NSC2Synthesis}, under As. \ref{As:PositiveDissipativity} and As. \ref{As:NegativeDissipativity}, Lm. \ref{Lm:AlternativeLMI_Schur} is applicable here to write an equivalent condition for \eqref{Eq:Pr:NSC4SynthesisStep1}. This condition, under the change of variables 
$
\scriptsize
\bm{L_{uy} & L_{u\bar{y}} & L_{uw} \\ L_{\bar{u}y} & L_{\bar{u}\bar{y}} & L_{\bar{u}w}} = 
\bm{\textbf{X}_p^{11} & \0 \\ \0 & \bar{\textbf{X}}_{\bar{p}}^{11}}
\bm{M_{uy} & M_{u\bar{y}} & M_{uw} \\ M_{\bar{u}y} & M_{\bar{u}\bar{y}} & M_{\bar{u}w}}
$, can be re-formulated as the LMI \eqref{Eq:Pr:NSC4Synthesis1}.
\end{proof}

\begin{figure*}[!hb]
\centering
\hrulefill
\begin{equation}\label{Eq:Pr:NSC2Synthesis2}
	\bm{
		\textbf{X}_p^{11} & \0 & L_{uy} & L_{uw} \\
		\0 & -\textbf{Y}^{22} & -\textbf{Y}^{22}M_{zy} & -\textbf{Y}^{22} M_{zw}\\ 
		L_{uy}^\T & -M_{zy}^\T\textbf{Y}^{22} & -L_{uy}^\T \textbf{X}^{12}-\textbf{X}^{21}L_{uy}-\textbf{X}_p^{22} & -\textbf{X}^{21}L_{uw}+M_{zy}^\T\textbf{Y}^{21} \\
		L_{uw}^\T & -M_{zw}^\T\textbf{Y}^{22} & -L_{uw}^\T \textbf{X}^{12}+\textbf{Y}^{12} M_{zy} &  M_{zw}^\T\textbf{Y}^{21} +  \textbf{Y}^{12}M_{zw} + \textbf{Y}^{11}
	}>0
\end{equation}
\begin{equation}\label{Eq:Pr:NSC3Synthesis2}
    \bm{
	\textbf{X}_p^{11} & \0 & L_{uy} & L_{u\bar{y}}\\ 
	\0 & \bar{\textbf{X}}_{\bar{p}}^{11} & L_{\bar{u}y} & L_{\bar{u}\bar{y}}\\
	L_{uy}^\T & L_{\bar{u}y}^\T & -L_{uy}^\T \textbf{X}^{12} - \textbf{X}^{21} L_{uy} -\textbf{X}_p^{22} & -\textbf{X}^{21}L_{u\bar{y}}-L_{\bar{u}y}^\T \bar{\textbf{X}}^{12} \\
	L_{u\bar{y}}^\T & L_{\bar{u}\bar{y}}^\T & -L_{u\bar{y}}^\T\textbf{X}^{12}-\bar{\textbf{X}}^{21} L_{\bar{u}y} & -L_{\bar{u}\bar{y}}^\T \bar{\textbf{X}}^{12} - \bar{\textbf{X}}^{21} L_{\bar{u}\bar{y}} - \bar{\textbf{X}}_{\bar{p}}^{22}
}>0,
\end{equation}
\begin{equation}\label{Eq:Pr:NSC4Synthesis2}
	\bm{
		\textbf{X}_p^{11} & \0 & \0 & L_{uy} & L_{u\bar{y}} & L_{uw} \\
		\0 & \bar{\textbf{X}}_{\bar{p}}^{11} & \0 & L_{\bar{u}y} & L_{\bar{u}\bar{y}} & L_{\bar{u}w}\\
		\0 & \0 & -\textbf{Y}^{22} & -\textbf{Y}^{22} M_{zy} & -\textbf{Y}^{22} M_{z\bar{y}} & \textbf{Y}^{22} M_{zw}\\
		L_{uy}^\T & L_{\bar{u}y}^\T & - M_{zy}^\T\textbf{Y}^{22} & -L_{uy}^\T\textbf{X}^{12}-\textbf{X}^{21}L_{uy}-\textbf{X}_p^{22} & -\textbf{X}^{21}L_{u\bar{y}}-L_{\bar{u}y}^\T \bar{\textbf{X}}^{12} & -\textbf{X}^{21}L_{uw} + M_{zy}^\T \textbf{Y}^{21} \\
		L_{u\bar{y}}^\T & L_{\bar{u}\bar{y}}^\T & - M_{z\bar{y}}^\T\textbf{Y}^{22} & -L_{u\bar{y}}^\T\textbf{X}^{12}-\bar{\textbf{X}}^{21}L_{\bar{u}y} & 		-(L_{\bar{u}\bar{y}}^\T \bar{\textbf{X}}^{12} + \bar{\textbf{X}}^{21}L_{\bar{u}\bar{y}}+\bar{\textbf{X}}_{\bar{p}}^{22}) & -\bar{\textbf{X}}^{21} L_{\bar{u}w} + M_{z\bar{y}}^\T \textbf{Y}^{21} \\ 
		L_{uw}^\T & L_{\bar{u}w}^\T & -M_{zw}^\T \textbf{Y}^{22}& -L_{uw}^\T\textbf{X}^{12}+\textbf{Y}^{12}M_{zy} & -L_{\bar{u}w}^\T\bar{\textbf{X}}^{12}+ \textbf{Y}^{12} M_{z\bar{y}} & M_{zw}^\T\textbf{Y}^{21} + \textbf{Y}^{12}M_{zw} + \textbf{Y}^{11}
	}>0
\end{equation}
\end{figure*}


\subsection{\textbf{Some General Remarks}}

The following remarks can be made on the synthesis techniques proposed for NSCs 1-4 respectively in Props. \ref{Pr:NSC1Synthesis}-\ref{Pr:NSC4Synthesis}.

\begin{figure*}[!hb]
\centering
\hrulefill
\begin{equation}\label{Eq:Pr:NSC2SynthesisOptimal2}
\bm{
\textbf{X}_p^{11} & \0 & L_{uy} & L_{uw} \\
\0 & \bar{\rho}\I & M_{zy} & M_{zw} \\ 
L_{uy}^\T & M_{zy}^\T & -\textbf{X}^{21}L_{uy} - L_{uy}^\T\textbf{X}^{12} -\textbf{X}_p^{22} & -\textbf{X}^{21}L_{uw}+\frac{1}{2}M_{zy}^\T \\
L_{uw}^\T & M_{zw}^\T & -L_{uw}^\T \textbf{X}^{12} + \frac{1}{2} M_{zy} &  \frac{1}{2}M_{zw}^\T +  \frac{1}{2}M_{zw} - \nu \I
}>0,
\end{equation}
\begin{equation}\label{Eq:Pr:NSC2Synthesis2Alternative} 
\begin{aligned}
\bm{
 -\textbf{Y}^{22} & -\textbf{Y}^{22}M_{zy} & -\textbf{Y}^{22}M_{zw} \\
 -M_{zy}^\T\textbf{Y}^{22} & -L_{uy}^\T \textbf{X}^{12} - \textbf{X}^{21}L_{uy} - \textbf{X}_p^{22} + \Theta^{11} & -\textbf{X}^{21}L_{uw} + M_{zy}^\T \textbf{Y}^{21} + \Theta^{12} \\
-M_{zw}^\T\textbf{Y}^{22} & -L_{uw}^\T \textbf{X}^{12} + \textbf{Y}^{12}M_{zy} + \Theta^{21} & M_{zw}^\T \textbf{Y}^{21} +  \textbf{Y}^{12}M_{zw} + \textbf{Y}^{11} +  \Theta^{22} 
}>0\\
\mbox{ with }
\bm{\Theta^{11} & \Theta^{12} \\ \Theta^{21} & \Theta^{22}} \equiv 
\alpha^2 \bm{\textbf{X}_p^{11} &  \textbf{X}_p^{11}\\\textbf{X}_p^{11} & \textbf{X}_p^{11}}
-\alpha \bm{
L_{uy}^\T + L_{uy} & 
L_{uy}^\T + L_{uw} \\ 
L_{uw}^\T + L_{uy} & 
L_{uw}^\T + L_{uw}
}, \mbox{ and }\alpha \in \R.
\end{aligned}
\end{equation}
\end{figure*}

\begin{remark}\textit{(Partial Topology Design)}\label{Rm:PartialTopologyDesign}
The interconnection matrix synthesizing techniques proposed in Props. \ref{Pr:NSC1Synthesis}-\ref{Pr:NSC4Synthesis}) can be used even when the interconnection matrix is partially known - as it will only reduce the number of variables in the corresponding LMI problems. In other words, the proposed techniques can be used not only to design an interconnection topology from scratch but also to fine-tune an existing interconnection topology (e.g., to determine coupling weights).    
\end{remark}

\begin{remark}\textit{(Optimal Synthesis)}\label{Rm:EstimationOFPassivityIndices}
The LMI problem \eqref{Eq:Pr:NSC2Synthesis} (and \eqref{Eq:Pr:NSC4Synthesis1}) proposed for the NSC 2 (and 4) in Props. \ref{Pr:NSC2Synthesis} (and \ref{Pr:NSC4Synthesis}) can be modified to address the problems of synthesizing the \emph{optimal} interconnection matrix $M$ that: minimizes the system stability index L2G($\gamma$) or maximizes passivity indices IFP($\nu$) and OFP($\rho$). For example, for the latter case, the modified LMI problem \eqref{Eq:Pr:NSC2Synthesis} takes the form: 
\begin{equation}\label{Eq:Pr:NSC2SynthesisOptimal1}
    \begin{aligned}
    \max_{\substack {L_{uy},L_{uw},M_{zy},M_{zw},\\ \nu, \rho, \{p_i: i\in\N_N\}}} & c_1 \nu - c_2 \bar{\rho} \\
    \mbox{Sub. To: }& \nu > 0, \bar{\rho} > 0, p_i > 0, \forall i\in\N_N, \eqref{Eq:Pr:NSC2SynthesisOptimal2},\\
    \end{aligned}
\end{equation}
where $c_1,c_2>0$ are some pre-selected cost coefficients, and $M_{uy} \triangleq (\textbf{X}_p^{11})^{-1} L_{uy}$, $M_{uw} \triangleq  (\textbf{X}_p^{11})^{-1} L_{uw}$ and $\rho = 1/\bar{\rho}$. Note that \eqref{Eq:Pr:NSC2SynthesisOptimal2} above has been obtained from \eqref{Eq:Pr:NSC2Synthesis2} via applying $\textbf{Y} = \scriptsize \bm{-\nu \I & \frac{1}{2}\I \\ \frac{1}{2} \I & -\rho \I}$, the congruence principle \cite{Bernstein2009} and the change of variables $\bar{\rho}=1/\rho$. 
\end{remark}

\begin{remark}\textit{(Estimating Stability/Passivity Indices)}
When the interconnection matrix is pre-defined, still, solving an optimal synthesis problem like \eqref{Eq:Pr:NSC2SynthesisOptimal1} (with some minor modifications) will reveal the stability/passivity indices of the networked system. The importance of such a knowledge (a dissipativity property) is evident from this paper itself. 
\end{remark}

\begin{remark}(Analysis via Synthesis)\label{Rm:AnalysisViaSynthesis}
When the interconnection matrix is predefined, all the synthesis problems formulated in Props. \ref{Pr:NSC1Synthesis}-\ref{Pr:NSC4Synthesis} (with some minor modifications) can still be used to analyze the corresponding networked systems. Note that the aforementioned ``minor modifications'' basically refer to having to reverse the used change of variables (e.g., replacing $L_{uy}$ in \eqref{Eq:Pr:NSC1Synthesis} with $\textbf{X}_p^{11}M_{uy}$). In essence, for the analysis of the NSCs, instead of Props. \ref{Pr:NSC1Stability}-\ref{Pr:NSC4Dissipativity} we can use Props. \ref{Pr:NSC1Synthesis}-\ref{Pr:NSC4Synthesis}. Note also that LMI conditions in Props. \ref{Pr:NSC1Synthesis}-\ref{Pr:NSC4Synthesis} are less convoluted than those in Props. \ref{Pr:NSC1Stability}-\ref{Pr:NSC4Dissipativity}. We will use this to our advantage later on when decentralizing the analysis results proposed in Props.  \ref{Pr:NSC1Stability}-\ref{Pr:NSC4Dissipativity}. 
\end{remark}

\begin{remark}\textit{(Single Subsystem Case: NSC 1)}\label{Rm:SingleSubsystemCaseNSC1}
Consider the NSC 1 with $N=1$ and a scalar interconnection matrix $M_{uy}=m_{uy}\I$ \eqref{Eq:NSC1Interconnection} where $m_{uy}\in\R$. Now, if the subsystem $\Sigma_1$ is L2G($\gamma_1$), As. \ref{As:PositiveDissipativity} holds and $\textbf{X}_p^{11}=p_1\gamma_1^2\I$, $\textbf{X}_p^{22}=-p_1\I$, $\textbf{X}^{12} = \textbf{X}^{12} = \0$. Thus, the LMI condition \eqref{Eq:Pr:NSC1Synthesis} in Prop. \ref{Pr:NSC1Synthesis} can be written as  
$
    \scriptsize \bm{ p_1\gamma_1^2\I &  l_{uy}\I \\ l_{uy}^\T \I& p_1\I} \normalsize >0 \iff 
    \scriptsize \bm{ p_1\gamma_1^2 &  l_{uy} \\ l_{uy} & p_1} \normalsize >0 \iff 
    \{p_1 > 0, p_1^2 \gamma_1^2 - l_{uy}^2 > 0 \} \impliedby \{p_1 = k l_{uy}/\gamma_1, l_{uy}>0, k>1\} 
$. Now, using the change of variables relationships, we get $m_{uy}=l_{uy}/(p_1\gamma_1^2)=1/k\gamma_1 \iff m_{uy}\gamma_1 = 1/k < 1$, which is the well-known small-gain condition for stability.
\end{remark}

\begin{remark}\textit{(Single Subsystem Case: NSC 2)}\label{Rm:SingleSubsystemCaseNSC2}
Consider the NSC 2 with $N=1$ and a scalar interconnection matrix $M = \scriptsize \bm{m_{uy}\I & m_{uw}\I \\ m_{zy}\I & m_{zw}\I}$ \eqref{Eq:NSC2Interconnection} where $m_{uy},m_{uw},m_{zy},m_{zw}\in\R$. Now, if the subsystem $\Sigma_1$ is L2G($\gamma_1$), As. \ref{As:PositiveDissipativity} holds and $\textbf{X}_p^{11}=p_1\gamma_1^2\I$, $\textbf{X}_p^{22}=-p_1\I$, $\textbf{X}^{12}=\textbf{X}^{21}=\0$. Let the dissipativity specification for the NSC 2 be $\textbf{Y}=\scriptsize\bm{-\nu_1\I & \frac{1}{2}\I \\ \frac{1}{2}\I & -\rho_1\I}$ with some $\nu_1, \rho_1 > 0$. Hence the As. \ref{As:NegativeDissipativity} also holds and the LMI problem \eqref{Eq:Pr:NSC2Synthesis} in Prop. \ref{Pr:NSC2Synthesis} now can be written as  
\begin{equation}\label{Eq:Pr:NSC2SynthesisSingleSubsystemCase}
    \begin{aligned}
    \mbox{Find: }& l_{uy}, l_{uw}, m_{zy}, m_{zw}, \{p_i: i\in\N_N\}\\
    \mbox{ Sub. to: }& p_i > 0, \forall i\in\N_N,\\
    &\hspace{-2mm}\bm{
    p_1\gamma_1^2 & 0 & l_{uy} & l_{uw} \\
    0 & \rho_1 & \rho_1 m_{zy} & \rho_1 m_{zw}\\ 
    l_{uy} & \rho_1 m_{zy} & p_1 & \frac{1}{2}m_{zy} \\
    l_{uw} & \rho_1 m_{zw} & \frac{1}{2} m_{zy} &  m_{zw} - \nu_1 
    }>0,
    \end{aligned}
\end{equation}
where $m_{uy}=l_{uy}/(p_1\gamma_1^2)$ and $m_{uw}=l_{uw}/(p_1\gamma_1^2)$. It can be shown that the above LMI problem is a streamlined version of the interconnection matrix synthesis approach proposed for a single subsystem in \cite[Th. 3]{Xia2014}.
\end{remark}

\begin{remark}(Assumption \ref{As:PositiveDissipativity}) \label{Rm:As:PositiveDissipativityFailure}
For situations where As. \ref{As:PositiveDissipativity} does not hold, a conservative solution (a fix) was proposed in Rm. \ref{Rm:As:PositiveDissipativity}. However, there is an alternative (and less conservative) solution to this issue based on Lm. \ref{Lm:AlternativeLMI_LowerBound}. To illustrate this solution, consider the interconnection matrix synthesis problem for the NSC 2 (originally addressed in Prop. \ref{Pr:NSC2Synthesis}). Now, without loss of generality, assume each subsystem $\Sigma_i, i \in  \N_N$ to be $X_i$-EID with $X_i^{11}<0$. Consequently, $X_p^{11} < 0$ and the corresponding LMI condition \eqref{Eq:Pr:NSC2Synthesis2} in  Prop. \ref{Pr:NSC2Synthesis} clearly does not hold. This can be resolved soon-after the step \eqref{Eq:Pr:NSC2SynthesisStep1} in the proof of Prop. \ref{Pr:NSC2Synthesis} by using both Lm. \ref{Lm:AlternativeLMI_LowerBound} (to handle $X_p^{11} < 0$) and Lm. \ref{Lm:AlternativeLMI_Schur} (to handle $-\textbf{Y}^{22}>0$) - which leads to the LMI condition \eqref{Eq:Pr:NSC2Synthesis2Alternative}. In all, when As. \ref{As:PositiveDissipativity} does not hold, the LMI condition \eqref{Eq:Pr:NSC2Synthesis2} in Prop. \ref{Pr:NSC2Synthesis} should be replaced by \eqref{Eq:Pr:NSC2Synthesis2Alternative}. Using the same steps, for situations where As. \ref{As:PositiveDissipativity} does not hold, similar solutions can be developed to the other proposed synthesis techniques in Props. \ref{Pr:NSC1Synthesis}, \ref{Pr:NSC3Synthesis} and \ref{Pr:NSC4Synthesis} as well.
\end{remark}

\begin{remark}
(Equilibrium Independence)
Similar to the proposed analysis techniques in Props. \ref{Pr:NSC1Stability}-\ref{Pr:NSC4Dissipativity}, the proposed synthesis techniques in Props. \ref{Pr:NSC1Synthesis}-\ref{Pr:NSC4Synthesis} are also independent of the equilibrium points of the respective NSCs 1-4. 
\end{remark}

\section{Decentralized Analysis and Synthesis of Networked Systems}\label{Sec:DecentralizedAnalysisAndSynthesis}

\begin{figure*}[!hb]
\centering
\hrulefill
\begin{equation}\label{Eq:Th:NSC1Stability2}
W_{ij} = \bm{  p_iX_i^{11}e_{ij} &  p_iX_i^{11} M_{uy}^{ij} \\ 
p_jM_{uy}^{ji\T}X_j^{11} &  
-p_jM^{ji\T}_{uy}X_j^{12} - p_iX_i^{21}M^{ij}_{uy} - p_iX_i^{22}e_{ij}}
\end{equation}
\begin{equation}\label{Eq:Th:NSC1Synthesis2}
W_{ij} = \bm{ p_iX_i^{11}e_{ij} &  L_{uy}^{ij} \\ L_{uy}^{ji\T} & -(p_jL_{uy}^{ji^\T}X_j^{12} + p_iX_i^{21}L_{uy}^{ij} + p_iX_i^{22}e_{ij})}
\end{equation}
\end{figure*}

In this section, we briefly show that the each centralized analysis and synthesis technique proposed earlier in Props. \ref{Pr:NSC1Stability}-\ref{Pr:NSC4Synthesis} for the NSCs 1-4 can be implemented in a decentralized and compositional manner. For this purpose, we will use the Lemmas \ref{Lm:NetworkMatrixProperties}, \ref{Lm:AlternativeLMI_BEW} and \ref{Lm:MainLemma}. First, to simplify the decentralization, we make the following assumption regarding the subsystems.

\begin{assumption}\label{As:SymmetryForDecentralization}
In the NSC 2, the input signal $w \equiv [w_i^\T]^\T_{i\in\N_N}$ where $w_i\in \R^{r_i}$ and the output signal $z \equiv [z_i^\T]^\T_{i\in\N_N}$ where $z_i\in\R^{l_i}$. 
Further, in the NSCs 3 and 4, $\bar{N}=N$ and each subsystem $\Sigma_i,i\in\N_N$ has a corresponding (virtual twin) subsystem $\bar{\Sigma}_i,i\in\N_{\bar{N}}$. Furthermore, in the NSC 4, the input signal $w \equiv [w_i^\T,\bar{w}_i^\T]^\T_{i\in\N_N}$ where $\bar{w}_i\in \R^{\bar{r}_i}$ and the output signal $z \equiv [z_i^\T, \bar{z}_i^\T]^\T_{i\in\N_N}$ where $\bar{z}_i\in\R^{\bar{l}_i}$. 
\end{assumption}

Note that under As. \ref{As:SymmetryForDecentralization}, each block element of any interconnection matrix $M$ in \eqref{Eq:NSC1Interconnection}-\eqref{Eq:NSC4Interconnection} is a block network matrix. For example, the block element $M_{uy} \equiv  [M_{uy}^{ij}]_{i,j\in\N_N}$ of $M$ (appearing in  \eqref{Eq:NSC1Interconnection}-\eqref{Eq:NSC4Interconnection}) is a block network matrix. Note that the matrix $M_{uy}^{ij}$ represents how the input $u_i$ of the subsystem $\Sigma_i$ is affected by the output $y_j$ of the subsystem $\Sigma_j$.

Moreover, under As. \ref{As:SymmetryForDecentralization}, each block element of the matrices $\textbf{X}_p$ \eqref{Eq:Pr:NSC1Stability2}, $\bar{\textbf{X}}_{\bar{p}}$ \eqref{Eq:Pr:NSC3Stability2} and $\textbf{Y}$ \eqref{Eq:Pr:NSC2Dissipativity0} are block network matrices. For example, the $(k,l)$\tsup{th} block elements of $\textbf{X}_p$, $\textbf{Y}$ and $\bar{X}_{\bar{p}}, \forall k,l\in\N_2$ are the block network matrices $\textbf{X}_p^{kl}\equiv \diag(p_i X_i^{kl}:i\in\N_N)$, $\bar{\textbf{X}}_{\bar{p}}^{kl}=\diag(\bar{p}_i \bar{X}_i^{kl}:i\in\N_N)$ and $\textbf{Y}^{kl}$, respectively. The matrices $\textbf{X}^{12},\textbf{X}^{21},\bar{\textbf{X}}^{12}$ and $\bar{\textbf{X}}^{21}$ are also block network matrices (recall: $\textbf{X}^{12} \triangleq \diag((X_i^{11})^{-1}X_i^{12}:i\in\N_N)$). 

It is worth noting that all the network matrices identified above are also block diagonal except for the block elements of $M$ and $\textbf{Y}$ - which are by default non-block diagonal (i,e., general) network matrices unless specified otherwise.

As pointed out in Rm. \ref{Rm:Lm:MainLemma}, in a network setting, Lm. \ref{Lm:MainLemma} can be used to decentrally analyze (test/enforce) an LMI condition of the form $W>0$ when $W$ is a network matrix. However, if we consider the main LMI conditions derived in Props. \ref{Pr:NSC1Synthesis}-\ref{Pr:NSC4Synthesis} (e.g., see \eqref{Eq:Pr:NSC2Synthesis2}), while they take the form $\Psi>0$, $\Psi$, in any of these cases, is not a network matrix. In fact, in each of these four cases, $\Psi$ is a block-block matrix. Moreover, upon close examination using Lm. \ref{Lm:NetworkMatrixProperties}, it can be seen that each block element of each $\Psi$ is a network matrix (assuming $\textbf{Y}^{12},\textbf{Y}^{21}$ and $\textbf{Y}^{22}$ to be block diagonal network matrices). Therefore, using Lm. \ref{Lm:AlternativeLMI_BEW}, we can replace the condition $\Psi>0$ with an equivalent condition $W \triangleq \text{BEW}(\Psi)>0$. Now, based on Lm. \ref{Lm:NetworkMatrixProperties}, $W$ is a network matrix. In all, using Lemmas \ref{Lm:NetworkMatrixProperties} and \ref{Lm:AlternativeLMI_BEW}, the main LMI conditions derived in Props. \ref{Pr:NSC1Synthesis}-\ref{Pr:NSC4Synthesis} can be transformed to the form $W>0$ where $W$ is a network matrix. Consequently, Props. \ref{Pr:NSC1Synthesis}-\ref{Pr:NSC4Synthesis} (centralized synthesis techniques) can be decentrally implemented exploiting Lm. \ref{Lm:MainLemma}. 

Note also that, the same argument is valid for the scenarios where As. \ref{As:PositiveDissipativity} is violated in Props. \ref{Pr:NSC1Synthesis}-\ref{Pr:NSC4Synthesis} (due to Rm. \ref{Rm:As:PositiveDissipativityFailure}). Moreover, as pointed out in Rm. \ref{Rm:AnalysisViaSynthesis},  Props. \ref{Pr:NSC1Synthesis}-\ref{Pr:NSC4Synthesis} can also be used to analyze the corresponding networked systems. Therefore, by extension, such centralized analysis techniques can also be decentrally implemented exploiting Lm. \ref{Lm:MainLemma}.

Based on the above discussion/proof, we now can formally state the decentralized versions of all the centralized analysis and synthesis techniques proposed in Props. \ref{Pr:NSC1Stability}-\ref{Pr:NSC4Synthesis} (as a set of eight theorems). However, due to space constraints, here we limit only to the NSC 1 and provide the decentralized versions of its centralized analysis and synthesis techniques proposed in Pr. \ref{Pr:NSC1Stability} and \ref{Pr:NSC1Synthesis} in the following two theorems.

\begin{theorem}\label{Th:NSC1Stability}
Under As. \ref{As:PositiveDissipativity} and \ref{As:SymmetryForDecentralization}, the NSC 1 is stable if at each subsystem $\Sigma_i, i\in\N_N$, the LMI problem 
\begin{equation}\label{Eq:Th:NSC1Stability1}
    \begin{aligned}
    \mbox{Find: } &p_i \\
    \mbox{ such that } &p_i>0, \ \tilde{W}_{ii}>0, 
    \end{aligned}
\end{equation}
is feasible, where $\tilde{W}_{ii}$ is computed from Lm. \ref{Lm:MainLemma} when analyzing $W=[W_{ij}]_{i,j\in\N_N}>0$ with $W_{ij}$ as in \eqref{Eq:Th:NSC1Stability2}. 
\end{theorem}

\begin{theorem}\label{Th:NSC1Synthesis}
Under As. \ref{As:PositiveDissipativity} and \ref{As:SymmetryForDecentralization}, the NSC 1 is stable if at each subsystem $\Sigma_i, i\in\N_N$, the problem
\begin{equation}\label{Eq:Th:NSC1Synthesis1}
\begin{aligned}
    \mbox{Find } &p_i, \{L_{uy}^i\} \\
    \mbox{ such that } &p_i > 0, \tilde{W}_{ii}>0,
\end{aligned}
\end{equation}
is feasible, where $\tilde{W}_{ii}$ is computed from Lm. \ref{Lm:MainLemma} when analyzing $W=[W_{ij}]_{i,j\in\N_N}>0$ with $W_{ij}$ as in \eqref{Eq:Th:NSC1Synthesis2}. The newly synthesized interconnection matrix elements at subsystem $\Sigma_i$, i.e. $\{M_{uy}^i\}$ can be found using the relationships:
$M_{uy}^{ij} = (p_iX_i^{11})^{-1}L_{uy}^{ij}, \forall j\in\N_{i-1}$, 
$M_{uy}^{ji} = (p_jX_j^{11})^{-1}L_{uy}^{ji}, \forall j\in\N_{i-1}$, and 
$M_{uy}^{ii} = (p_iX_i^{11})^{-1}L_{uy}^{ii}$.
\end{theorem}

\section{Numerical Results}\label{Sec:NumericalResults}
In this section, we provide a simple numerical example to illustrate the theoretical results. We point out that this numerical example can also be seen as a ``networked'' version of the numerical example considered in \cite{Xia2014}. 

\subsection{Subsystems}

As the set of subsystems $\{\Sigma_i:i\in\N_5\}$ (in the NSCs 1-4), we consider a set of stable, non-passive, single-input-single-output (SISO) \emph{controllers} described by the transfer functions: \begin{equation}\label{Eq:NumResSubsystems}
    \Sigma_i: G_i(s) = \frac{a_is + b_i}{s+c_i}\exp(-d_i s)
\end{equation}
where 
$[a_i]_{i\in\N_5} \equiv [-2,-0.3,-0.2,-1.3,-1.2]$, 
$[b_i]_{i\in\N_5} \equiv [1,16,0.2,0.2,1]$, 
$[c_i]_{i\in\N_5} \equiv [1,9,1,1,3]$ and 
$[d_i]_{i\in\N_5} \equiv [1,1.2,0.8,0.9,1.1]$. Despite the involved delay component in \eqref{Eq:NumResSubsystems}, the $L_2$-gain values of these subsystems can be evaluated as 
$[\gamma^2_i]_{i\in\N_5}=[4.00, 3.16, 0.04, 1.69, 1.44]$. 

As the set of subsystems $\{\bar{\Sigma}_i:i\in\N_5\}$ (in the NSCs 3-4), we consider a set of unstable, non-passive, SISO \textbf{plants} described by the transfer functions:
\begin{equation}\label{Eq:NumResSubsystemsBar}
    \bar{\Sigma}_i: H_i(s) = \frac{a_i s + b_i}{s + c_i}
\end{equation}
where 
$[a_i]_{i\in\N_5} \equiv [1,1.5,2,0.5,3]$, 
$[b_i]_{i\in\N_5} \equiv [2,3,2,1,1.5]]$ and 
$[c_i]_{i\in\N_5} \equiv [-1,-3,-5,-2,-4]$. Due to the linear nature of \eqref{Eq:NumResSubsystemsBar}, output feedback passivity indices of these subsystems can be evaluated directly as $[\bar{\rho}_i]_{i\in\N_5}=[-0.50, -1.00, -2.50, -2.00, -2.67]$. 

It is worth noting that even if the subsystems are unknown and non-linear, there are many on-line as well as off-line techniques to directly obtain their dissipativity properties such as the $L_2$-gain values and the passivity indices (e.g., see \cite{Xia2014,Zakeri2019,WelikalaP42022,Arcak2022}).

\subsection{Network Topology}

Note that the known subsystem dissipativity properties (evaluated L2G and OFP indices given above) satisfy the As. \ref{As:PositiveDissipativity}, to synthesize the interconnection matrix $M$ in the NSCs 1-4, we can respectively use the LMIs given in Props. \ref{Pr:NSC1Synthesis}-\ref{Pr:NSC4Synthesis}. 

In particular, to guide the synthesis of the interconnection matrix $M$, we assume an underlying undirected \emph{graph topology} described respectively by the adjacency matrix $A=[A_{ij}]_{i,j\in\N_5}$ and a corresponding cost matrix $C=[C_{ij}]_{i,j\in\N_5}$:
\begin{equation}\label{Eq:NumResGraphTopology}
    A = \scriptsize \bm{0 & 1 & 1 & 0 & 0 \\ 1 & 0 & 1 & 0 & 0 \\ 1 & 1 & 0 & 1 & 1 \\ 0 & 0 & 1 & 0 & 0 \\ 0 & 0 & 1 & 0 & 0},\ \ \normalsize
    C= \scriptsize \bm{100 & 1 & 1 & 10 & 10 \\ 1 & 100 & 1 & 10 & 10 \\ 1 & 1 & 100 & 1 & 1 \\ 10 & 10 & 1 & 100 & 10 \\ 10 & 10 & 1 & 10 & 100}\normalsize.
\end{equation}
Note that $A_{ij}, i,j\in\N_5$ in \eqref{Eq:NumResGraphTopology} indicates whether an interconnection between subsystems $\Sigma_i$ and $\Sigma_j$ is allowed, while $C_{ij}$ in \eqref{Eq:NumResGraphTopology} represents the cost of such an interconnection.  

Now, when synthesizing the interconnection matrix $M$ (i.e., solving the LMIs given in Props. \ref{Pr:NSC1Synthesis}, \ref{Pr:NSC2Synthesis}, \ref{Pr:NSC3Synthesis} and \ref{Pr:NSC4Synthesis}), we have two options: (1) use the adjacency matrix $A$ \eqref{Eq:NumResGraphTopology} to impose a \emph{hard graph constraint} that completely restricts the usage of certain elements in $M$, or (2) use the cost matrix $C$ \eqref{Eq:NumResGraphTopology} to impose a \emph{soft graph constraint} that penalizes the usage of certain elements in $M$. Note that these hard and soft graph constraints can simply be included in the interested LMIs problems by: (1) introducing constraints of the form $M_{uy}^{ij}=0$ $\forall i,j\in\N_5$ such that $A_{ij}=0$, or (2) using a cost function of the form $\sum_{i,j\in\N_5}\vert  C_{ij}M_{uy}^{ij}\vert^2$, respectively.

\subsection{Observations}

\subsubsection{\textbf{NSC 1}}
First, to validate the Prop. \ref{Pr:NSC1Synthesis} and to illustrate the difference between the aforementioned hard and soft graph constraints, let us consider an example networked system of the form NSC 1 comprised of the subsystems in \eqref{Eq:NumResSubsystems} interconnected via an interconnection matrix $M_{uy}$ \eqref{Eq:NSC1Interconnection}. Figure \ref{Fig:NSC1Simulink} shows the Simulink implementation of this networked system. Note that, to perturb the initial state of the subsystems, we use a step signal as shown in Fig. \ref{Fig:NSC1w}. As shown in Fig. \ref{Fig:NSC1Case1y}, when $M_{uy}=\I$ is used (i.e., unit self loops), the networked system becomes unstable. This motivates the need to synthesize a stabilizing $M_{uy}$ using Prop. \ref{Pr:NSC1Synthesis}. Figures \ref{Fig:NSC1Case3y} and \ref{Fig:NSC1Case4y} show the observed output signals when the optimal $M_{uy}$ values:  
\begin{equation*}
    M_{uy} = 10^{-1} \times \scriptsize \bm{0 & 0.497 & 1.226 & 0 & 0 \\ 0.408 & 0 & 1.059 & 0 & 0 \\ -0.229 & -0.099 & 0 & 0.318 & 0.508 \\ 0 & 0 & 0.902 & 0 & 0 \\ 0 & 0 & 0.924 & 0 & 0}
\end{equation*}
and 
\begin{equation*}
    M_{uy} = \scriptsize \bm{0 & 0.099 & 0.043 & 0 & 0 \\ 0.110 & 0 & 0.048 & 0 & 0 \\ 0.967 & 0.966 & 0 & 1.510 & 1.510 \\ 0 & 0 & 0.094 & 0 & 0.011 \\ 0 & 0 & 0.104 & 0.012 & 0}
\end{equation*}
obtained under the hard and soft graph constraint are used, respectively. It is worth noting that when the soft graph constraint is used, an extra interconnection between the subsystems $\Sigma_5$ and $\Sigma_4$ is being used. In both cases, the synthesized $M_{uy}$ successfully stabilizes the networked system. 

Note that, in the remaining numerical examples, we have used hard graph constraints to restrict new interconnections and soft graph constraints to penalize self connections.

\begin{figure}[!hb]
    \centering
    \begin{subfigure}{0.48\textwidth}
    \centering
        \includegraphics[width=0.6\linewidth]{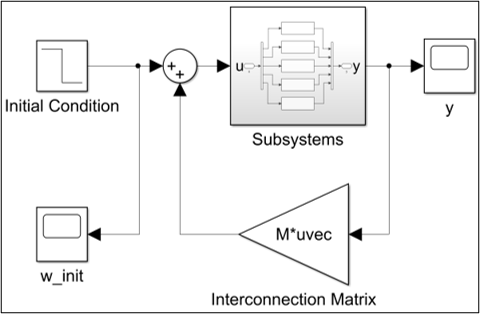}
        \caption{NSC 1: Simulink implementation using the subsystems \eqref{Eq:NumResSubsystems}.}
        \label{Fig:NSC1Simulink}
    \end{subfigure}\\
    \begin{subfigure}{0.23\textwidth}
        \includegraphics[width=\linewidth]{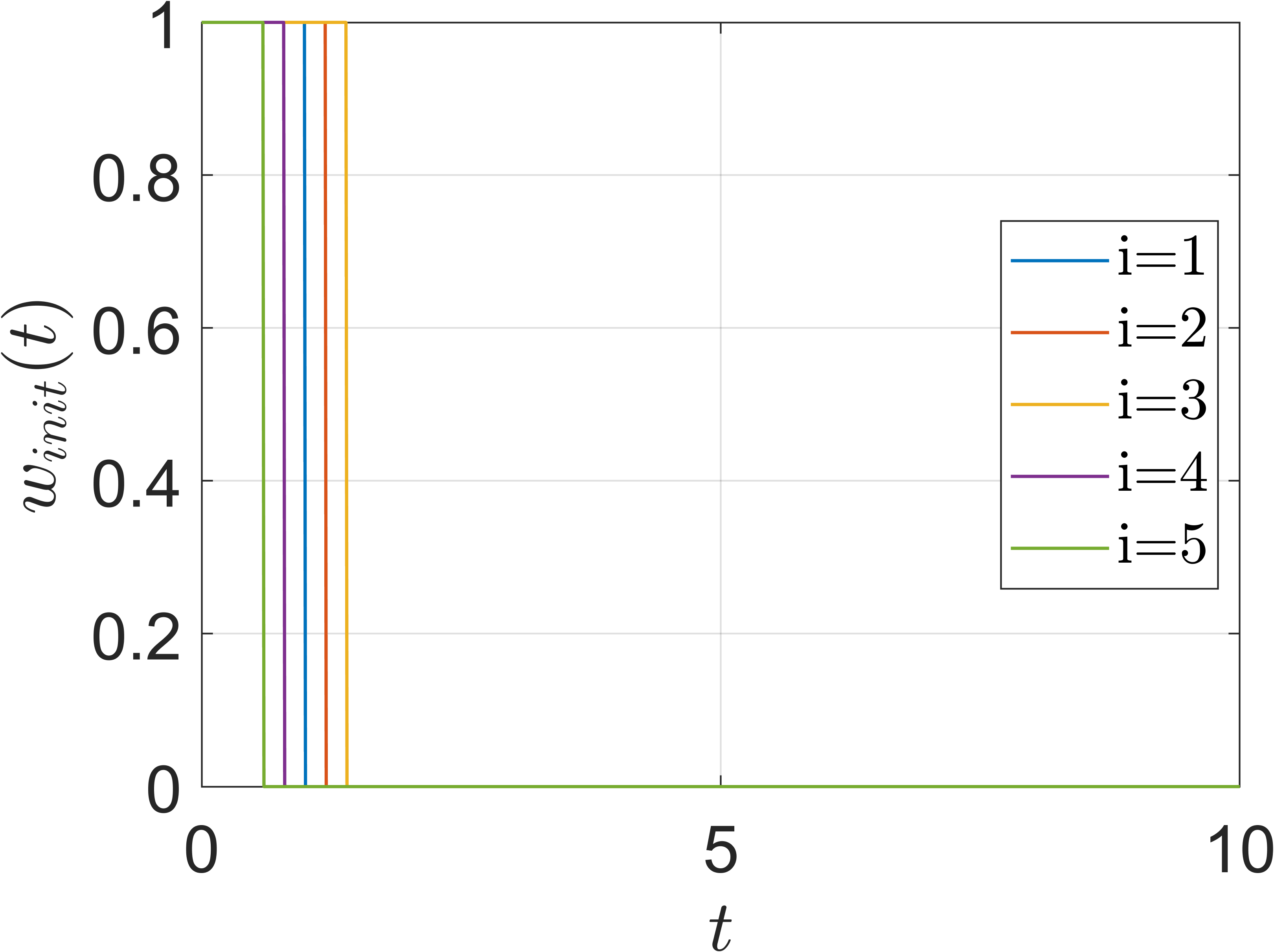}
        \caption{Initial excitation signal.}
        \label{Fig:NSC1w}
    \end{subfigure}
    \hfill
    \begin{subfigure}{0.23\textwidth}
        \includegraphics[width=\linewidth]{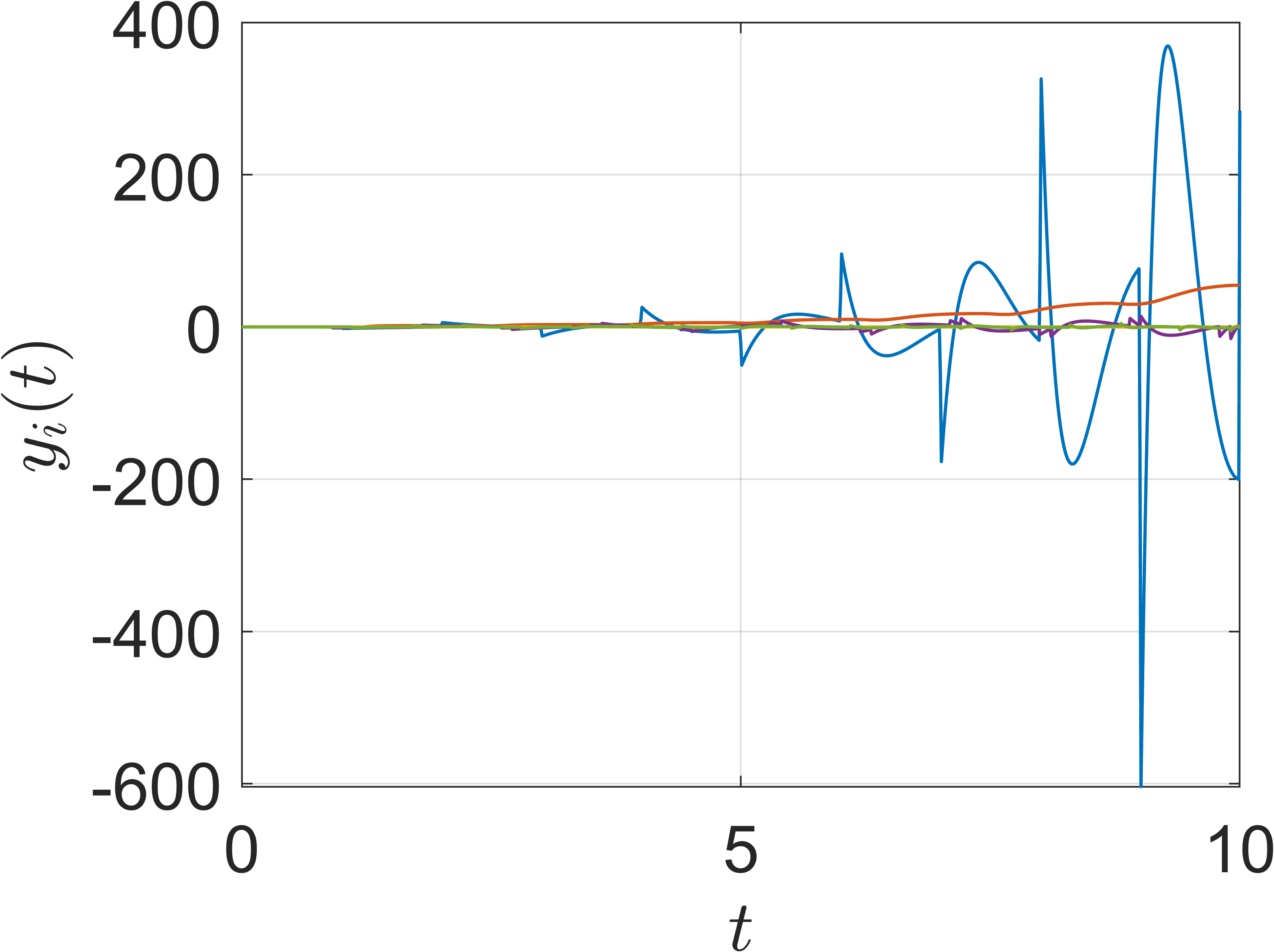}
        \caption{Output $y(t)$ with $M=\I$.}
        \label{Fig:NSC1Case1y}
    \end{subfigure}
    \begin{subfigure}{0.23\textwidth}
        \includegraphics[width=\linewidth]{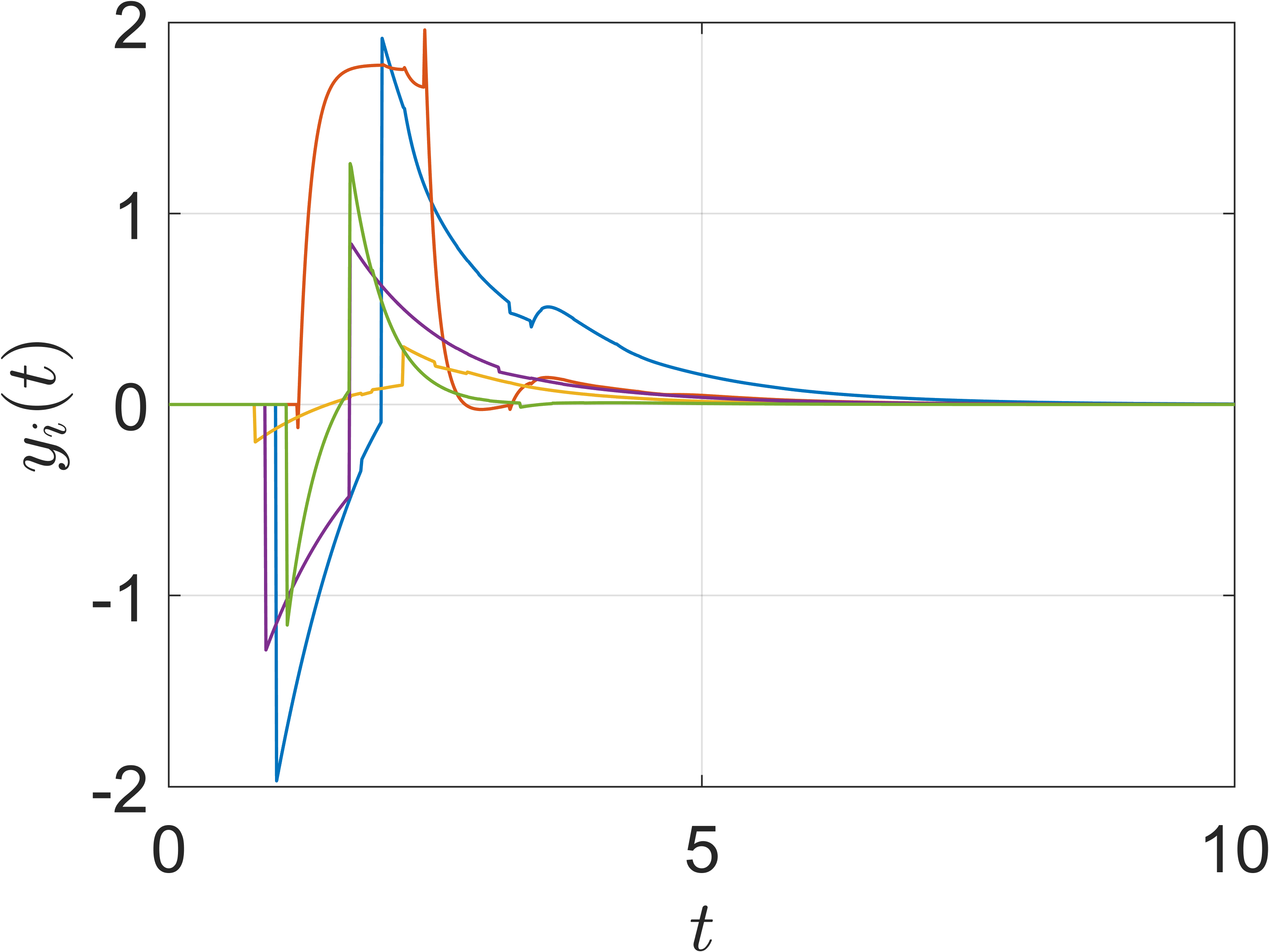}
        \caption{Output $y(t)$ with optimal $M$ under a hard graph constraint.}
        \label{Fig:NSC1Case3y}
    \end{subfigure}
    \hfill
    \begin{subfigure}{0.23\textwidth}
        \includegraphics[width=\linewidth]{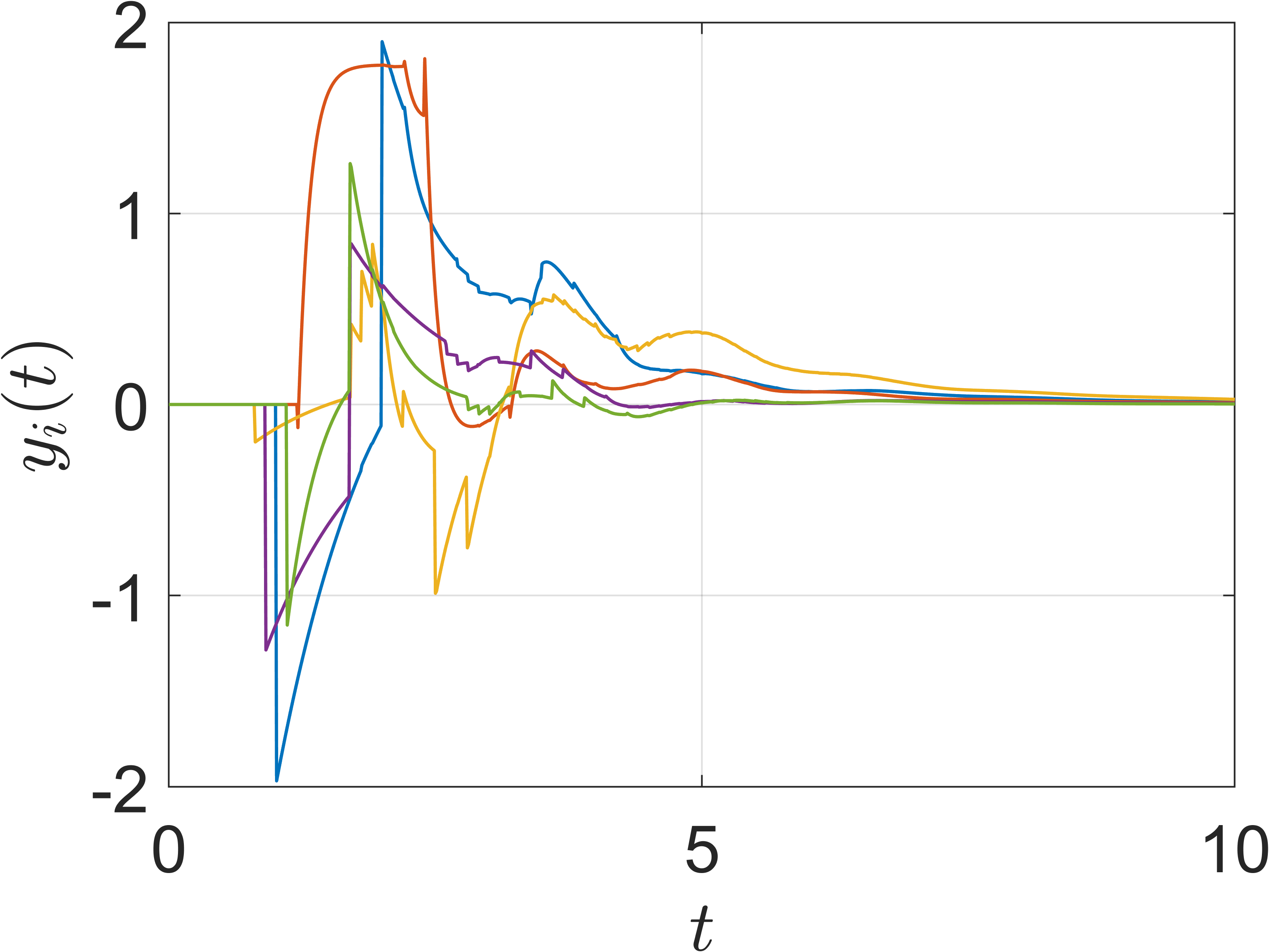}
        \caption{Output $y(t)$ with optimal $M$ under a soft graph constraint.}
        \label{Fig:NSC1Case4y}
    \end{subfigure}
    \caption{\textbf{NSC 1:} 
    \textbf{(a)} An example NSC 1 implemented in Simulink by interconnecting the controllers (subsystems) $\{\Sigma_i:i\in\N_5\}$ in \eqref{Eq:NumResSubsystems} via an interconnection matrix $M\in\R^{5\times 5}$.
    \textbf{(b)} The excitation signal used to perturb the subsystem initial conditions. 
    The observed output signals under: \textbf{(c)} $M=\I$, \textbf{(d)} optimally synthesized $M$ subject to a hard graph constraint, and \textbf{(e)} optimally synthesized $M$ subject to soft graph constraint.}
    \label{Fig:NSC1}
\end{figure}

\subsubsection{\textbf{NSC 2}}
To validate the Prop. \ref{Pr:NSC2Synthesis}, we next use the previous networked system (shown in Fig. \ref{Fig:NSC1Simulink}) with an added input port ($w$) and an output port ($z$) - making a networked system of the form NSC 2 shown in Fig. \ref{Fig:NSC2Simulink}. This networked system, as shown in Fig. \ref{Fig:NSC2Case1w}, is unstable under the interconnection matrix choice $M= \scriptsize \bm{\I & \I \\ \I & \I}$ \eqref{Eq:NSC2Interconnection}. Note also that it is desirable to make this networked system passive as it then can be used to control another non-passive networked system. This motivates the need to synthesize $M$ such that this networked system (shown in Fig. \ref{Fig:NSC2Simulink}) is maximally passive (from $w$ to $z$). For this purpose, we used Prop. \ref{Pr:NSC2Synthesis} (see also Rm. \ref{Rm:EstimationOFPassivityIndices}) and synthesized $M$ that optimized the input feedforward and output feedback passivity indices (i.e., $\nu$ and $\rho$, respectively) of the networked system. The obtained optimal passivity indices are as follows:
\begin{equation}
    \nu^* = 0 \ \ \mbox{ and } \ \ \rho^*=4.78.
\end{equation}
The output trajectories of the passivated networked system are shown in Figs. \ref{Fig:NSC2Case2y} and \ref{Fig:NSC2Case2z}. 

To further verify the obtained passivity measures of the networked system, we next consider the simple example scenario shown in Fig. \ref{Fig:NSC2ApplicationSimulink}. In there, the passivated networked system (obtained above) is connected in feedback with a memoryless system $K_{sys}\I$ where $K_{sys}\in\R$. According to \cite[Th. 4]{Xia2014}, the composite system is passive if $K_{sys}+\rho^*>0$, or non-passive otherwise. Figures \ref{Fig:NSC2Case3z} and \ref{Fig:NSC2Case4z} show output trajectories observed when $K_{sys}$ is selected such that $K_{sys}+\rho^* = 1$ (hence passive) and $K_{sys}+\rho^* = -1$ (hence non-passive), respectively. Their opposite nature implies the change of the passivity property of the composite system.

\begin{figure}[!t]
    \centering
    \begin{subfigure}{0.48\textwidth}
    \centering
        \includegraphics[width=0.75\linewidth]{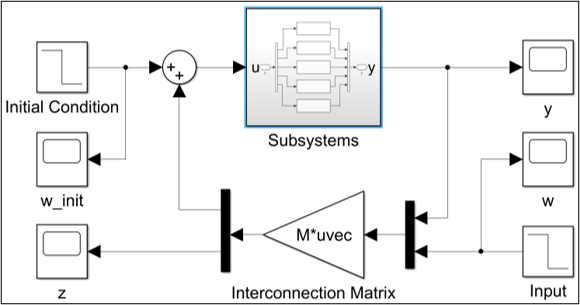}
        \caption{NSC 2: Simulink implementation using the subsystems \eqref{Eq:NumResSubsystems}.}
        \label{Fig:NSC2Simulink}
    \end{subfigure}\\
    \begin{subfigure}{0.23\textwidth}
        \includegraphics[width=\linewidth]{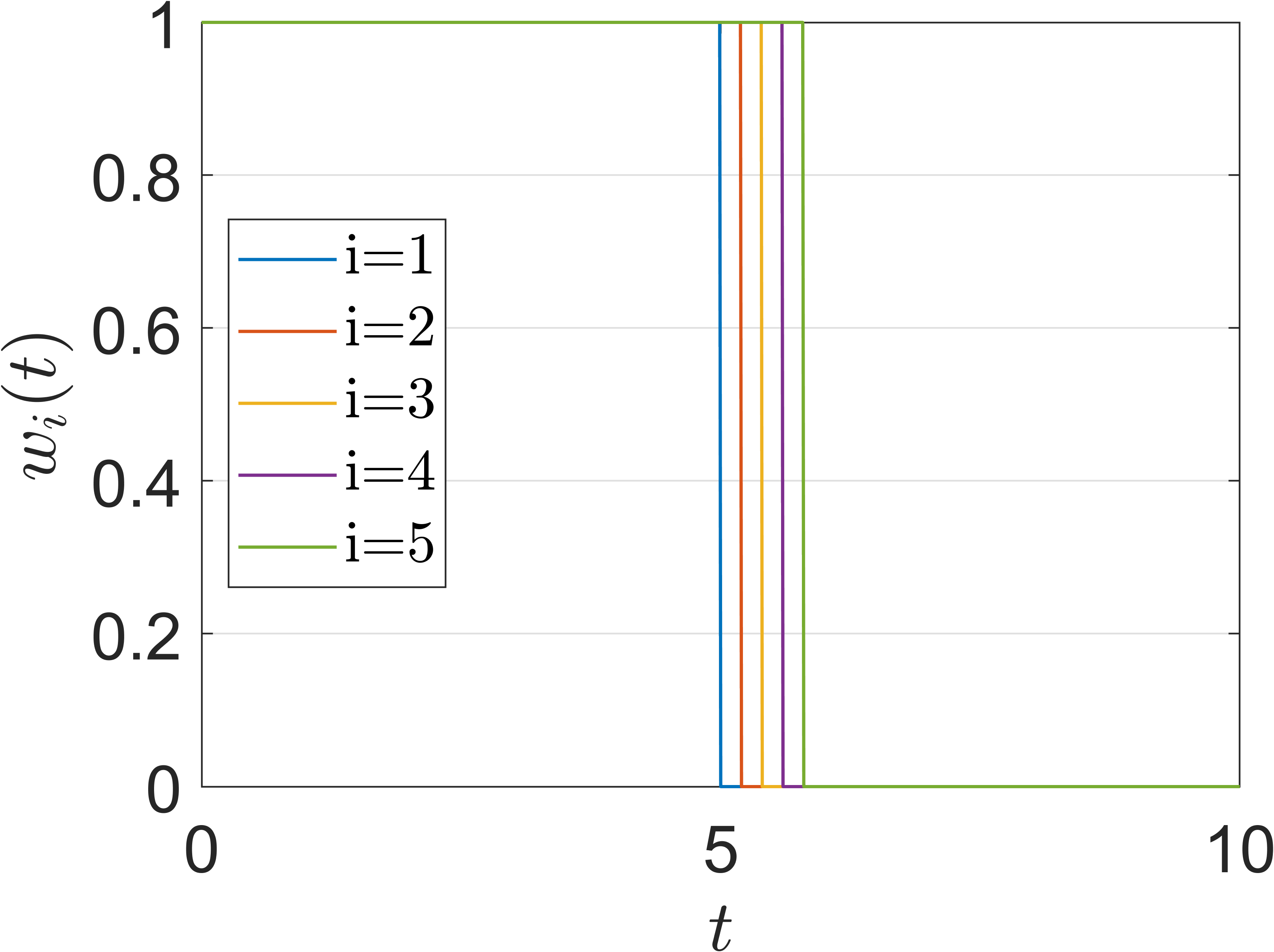}
        \caption{Exogenous input $w(t)$.}
        \label{Fig:NSC2Case1w}
    \end{subfigure}
    \hfill
    \begin{subfigure}{0.23\textwidth}
        \includegraphics[width=\linewidth]{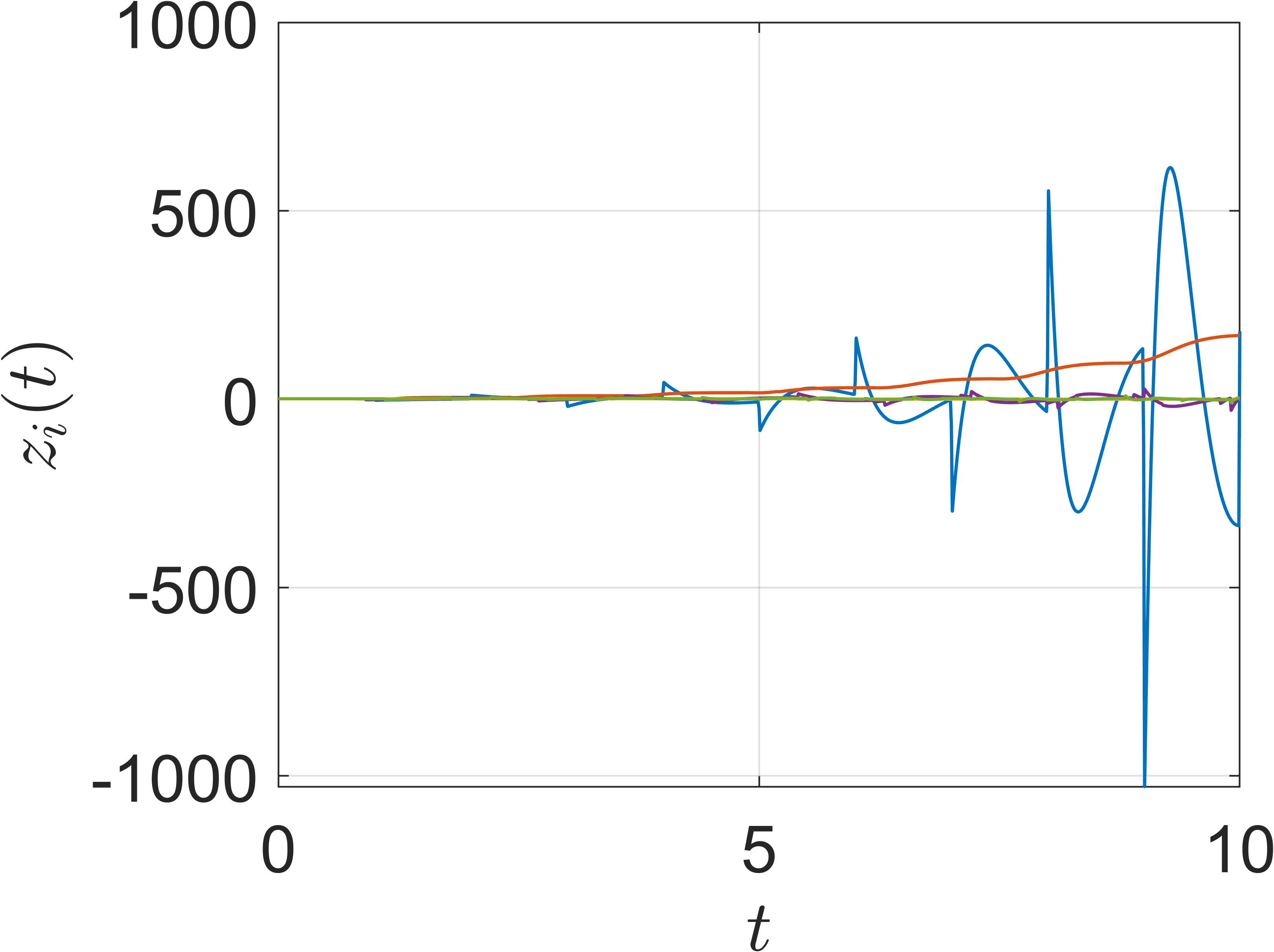}
        \caption{Output with $M= \scriptsize \bm{\I & \I \\ \I & \I}$.}
        \label{Fig:NSC2Case1z}
    \end{subfigure}
    \begin{subfigure}{0.23\textwidth}
        \includegraphics[width=\linewidth]{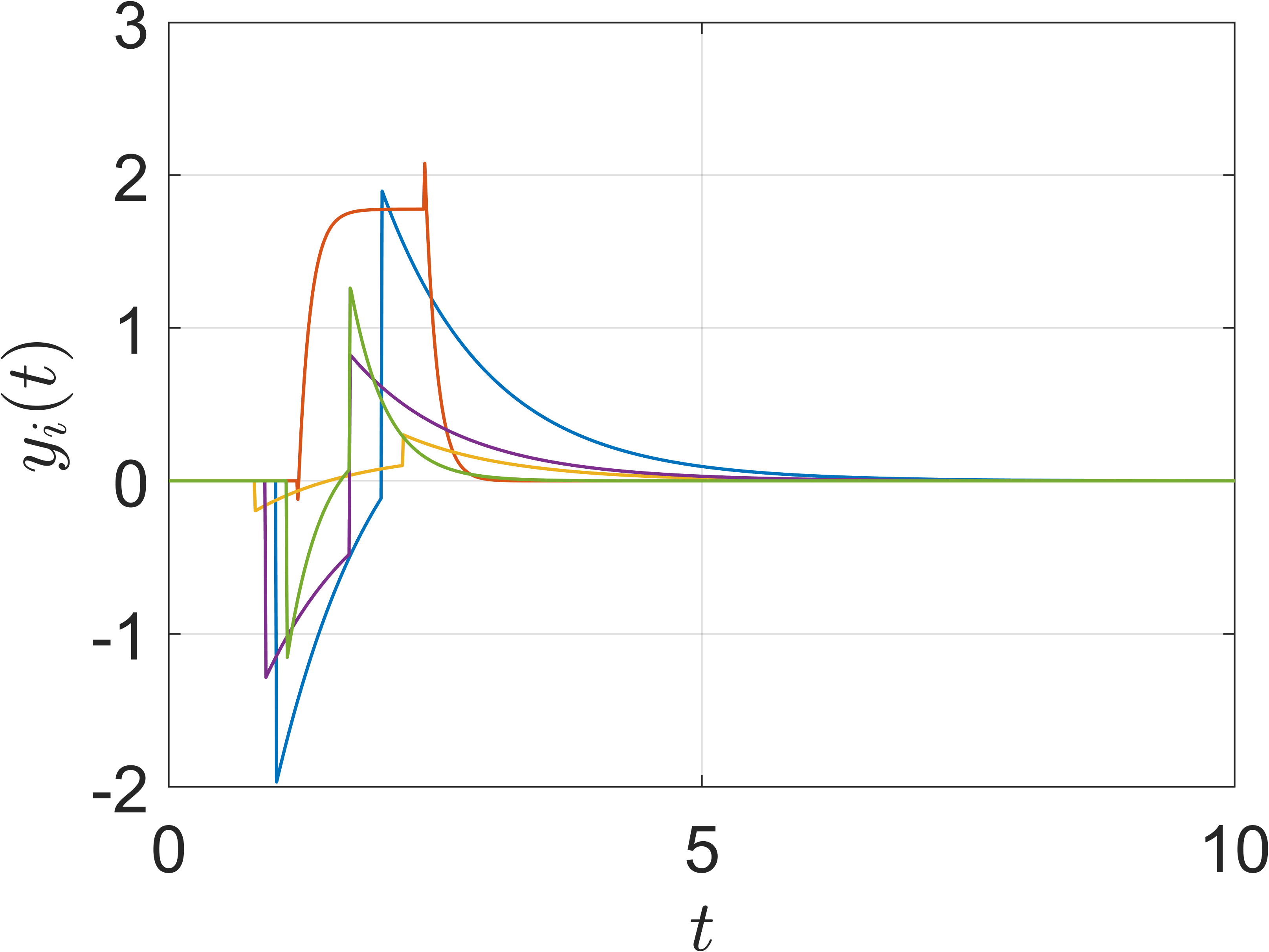}
        \caption{Subsystem outputs with optimal $M$.}
        \label{Fig:NSC2Case2y}
    \end{subfigure}
    \hfill
    \begin{subfigure}{0.23\textwidth}
        \includegraphics[width=\linewidth]{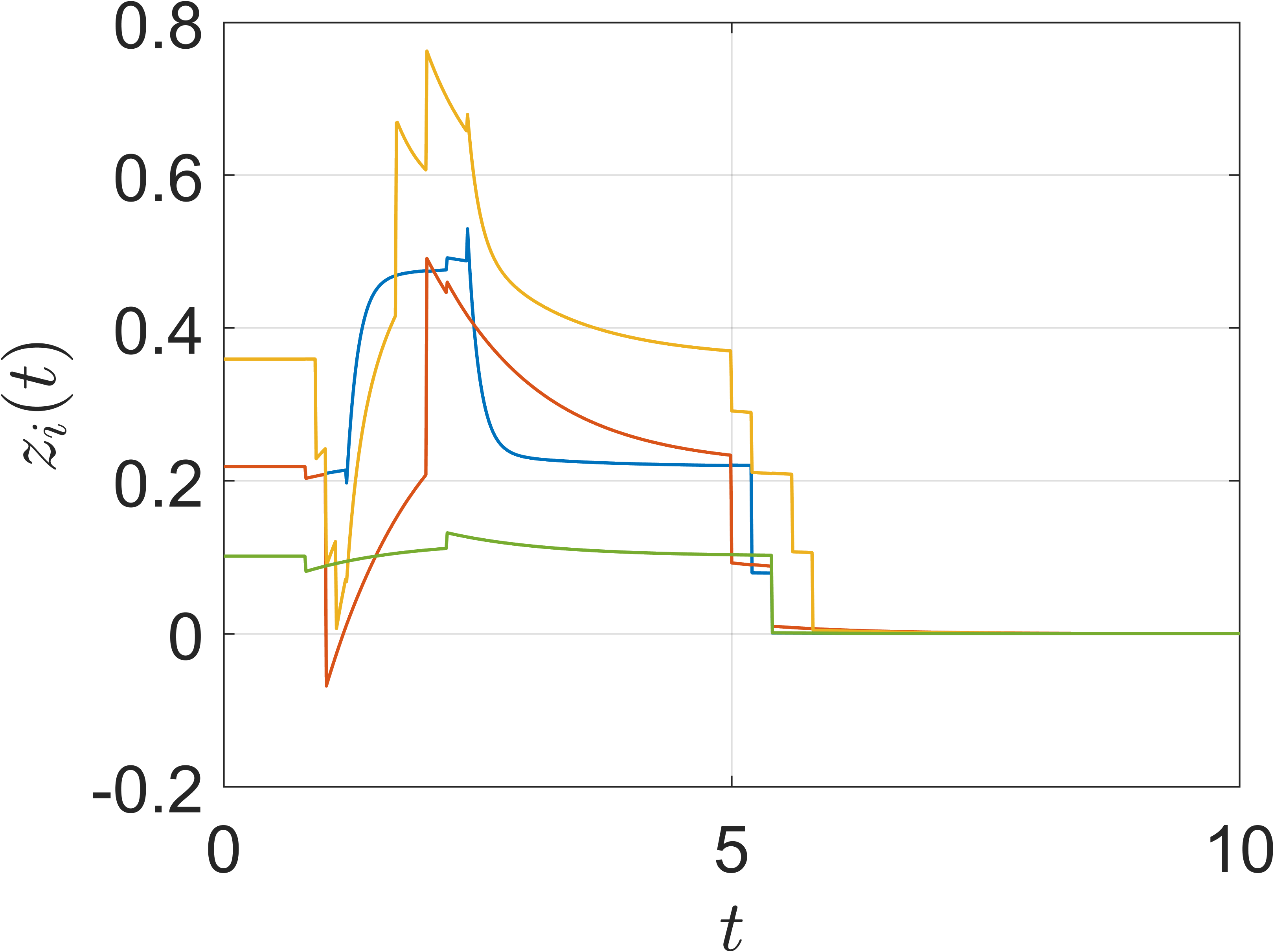}
        \caption{Output with optimal $M$.}
        \label{Fig:NSC2Case2z}
    \end{subfigure}
    \caption{\textbf{NSC 2:} \textbf{(a)} An example NSC 2 implemented in Simulink using the controllers \eqref{Eq:NumResSubsystems} and an interconnection matrix $M\in\R^{10 \times 10}$. \textbf{(b)} Used exogenous input signal (in addition to the initial excitation signal shown in Fig. \ref{Fig:NSC1w}). \textbf{(c)} Observed output under $M= \scriptsize \bm{\I & \I \\ \I & \I}$. Observed: \textbf{(d)} subsystem outputs and \textbf{(e)} the networked system output under optimally synthesized $M$.}
    \label{Fig:NSC2}
\end{figure}

\begin{figure}[!t]
    \centering
    \begin{subfigure}{0.48\textwidth}
    \centering
        \includegraphics[width=0.75\linewidth]{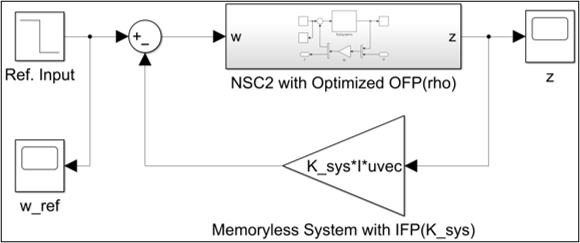}
        \caption{NSC 2 in Fig. \ref{Fig:NSC2} with a feedback connection.}
        \label{Fig:NSC2ApplicationSimulink}
    \end{subfigure}\\
    \begin{subfigure}{0.23\textwidth}
        \includegraphics[width=\linewidth]{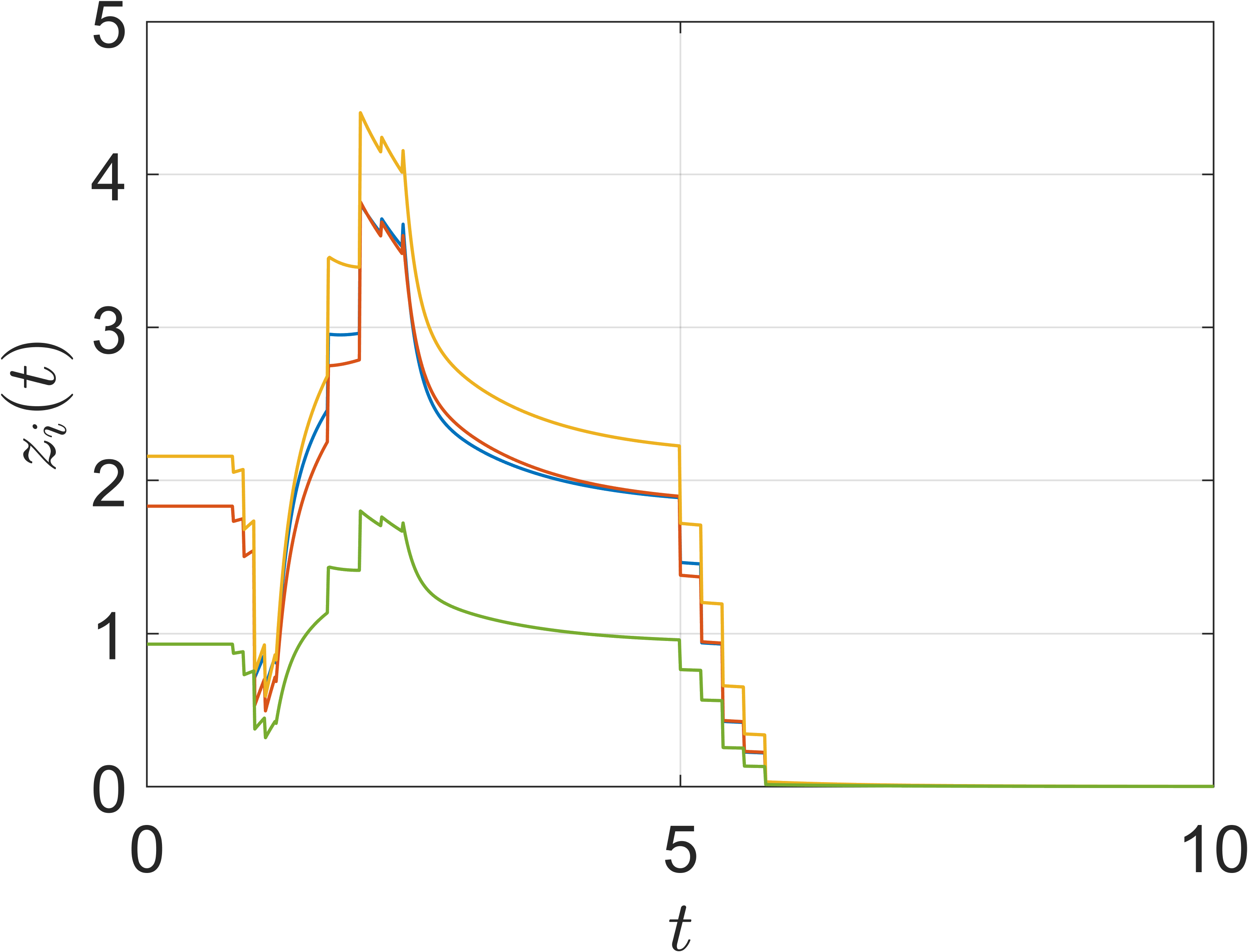}
        \caption{Output $z(t)$ with $K_{sys} = - \rho^* + 1$ (passive).}
        \label{Fig:NSC2Case3z}
    \end{subfigure}
    \hfill 
    \begin{subfigure}{0.23\textwidth}
        \includegraphics[width=\linewidth]{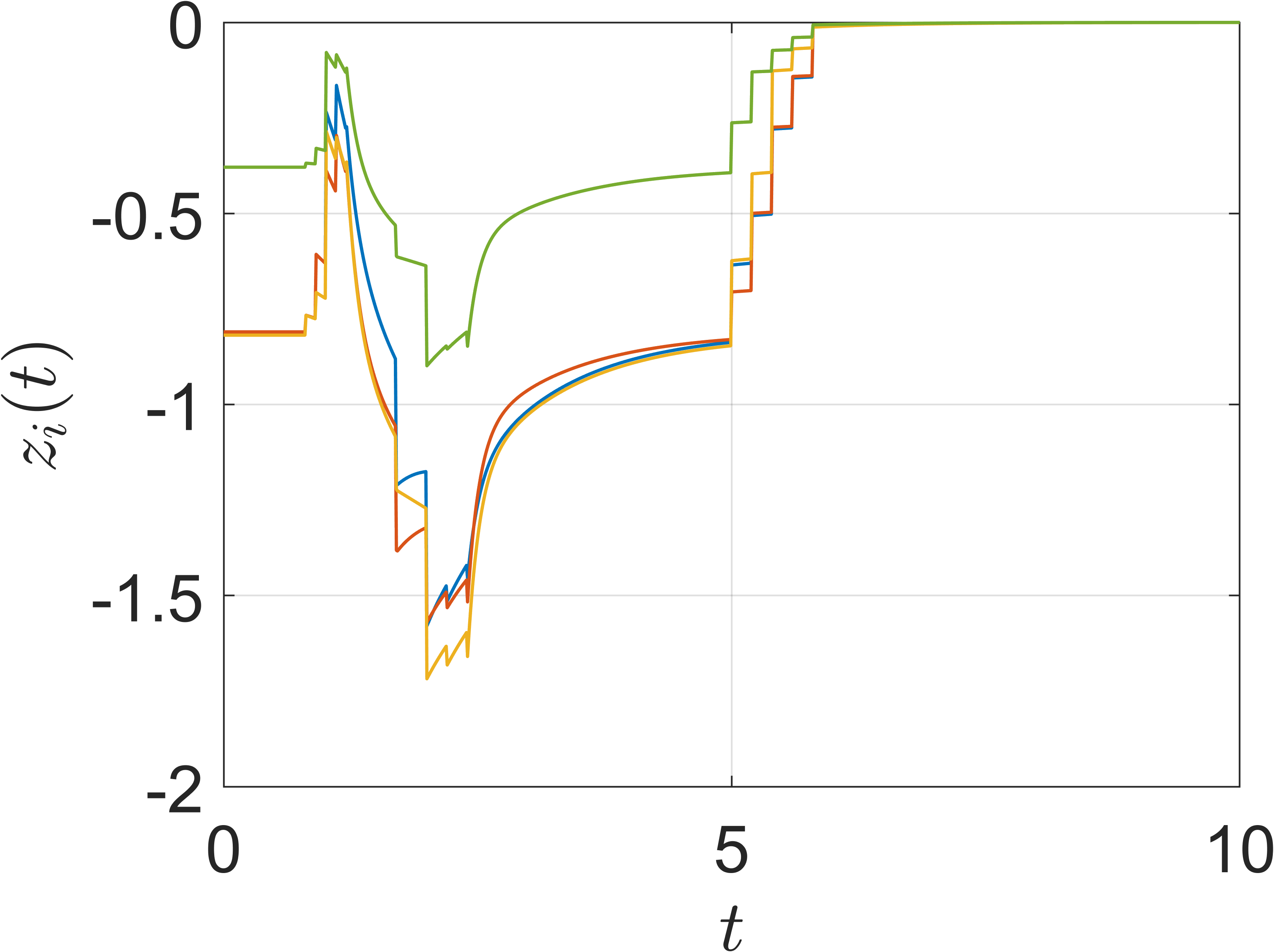}
        \caption{Output $z(t)$ with $K_{sys} = - \rho^* - 1$ (non-passive).}
        \label{Fig:NSC2Case4z}
    \end{subfigure}
    \caption{\textbf{NSC 2 Application:} A verification of the obtained passivity properties of the networked system.}
    \label{Fig:NSC2Application}
\end{figure}

\begin{figure}[!t]
    \centering
    \begin{subfigure}{0.48\textwidth}
    \centering
        \includegraphics[width=0.75\linewidth]{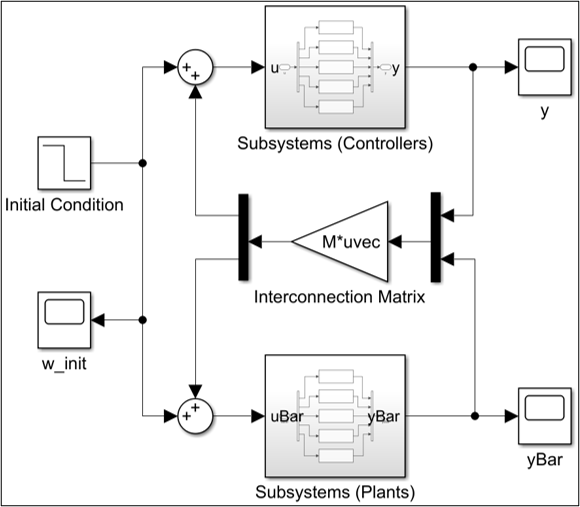}
        \caption{NSC 3: Simulink implementation using the subsystems \eqref{Eq:NumResSubsystems} (as controllers) and \ref{Eq:NumResSubsystemsBar} (as plants).}
        \label{Fig:NSC3Simulink}
    \end{subfigure}\\
    \begin{subfigure}{0.23\textwidth}
        \includegraphics[width=\linewidth]{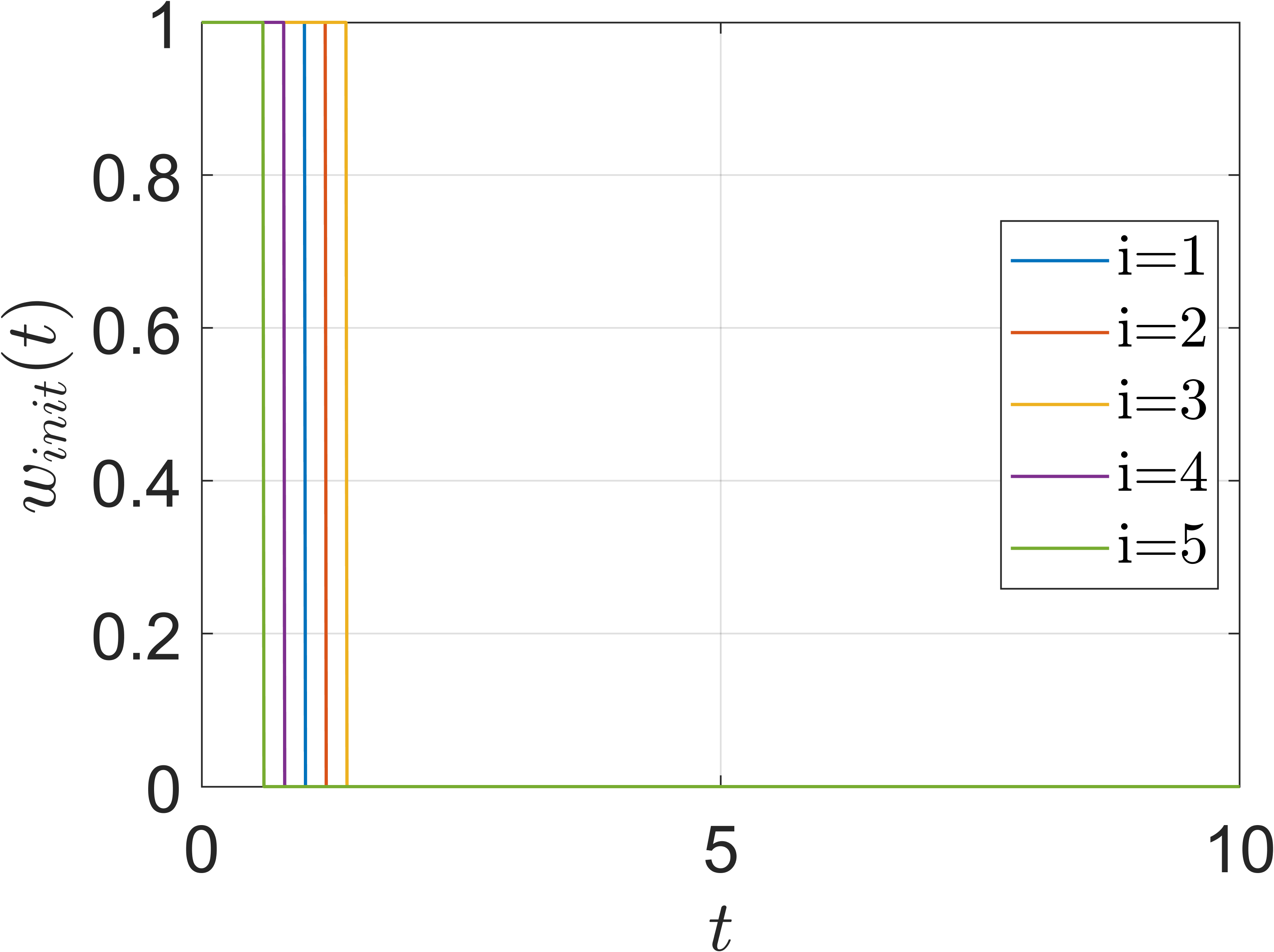}
        \caption{Initial excitation signal.}
        \label{Fig:NSC3w}
    \end{subfigure}
    \begin{subfigure}{0.23\textwidth}
        \includegraphics[width=\linewidth]{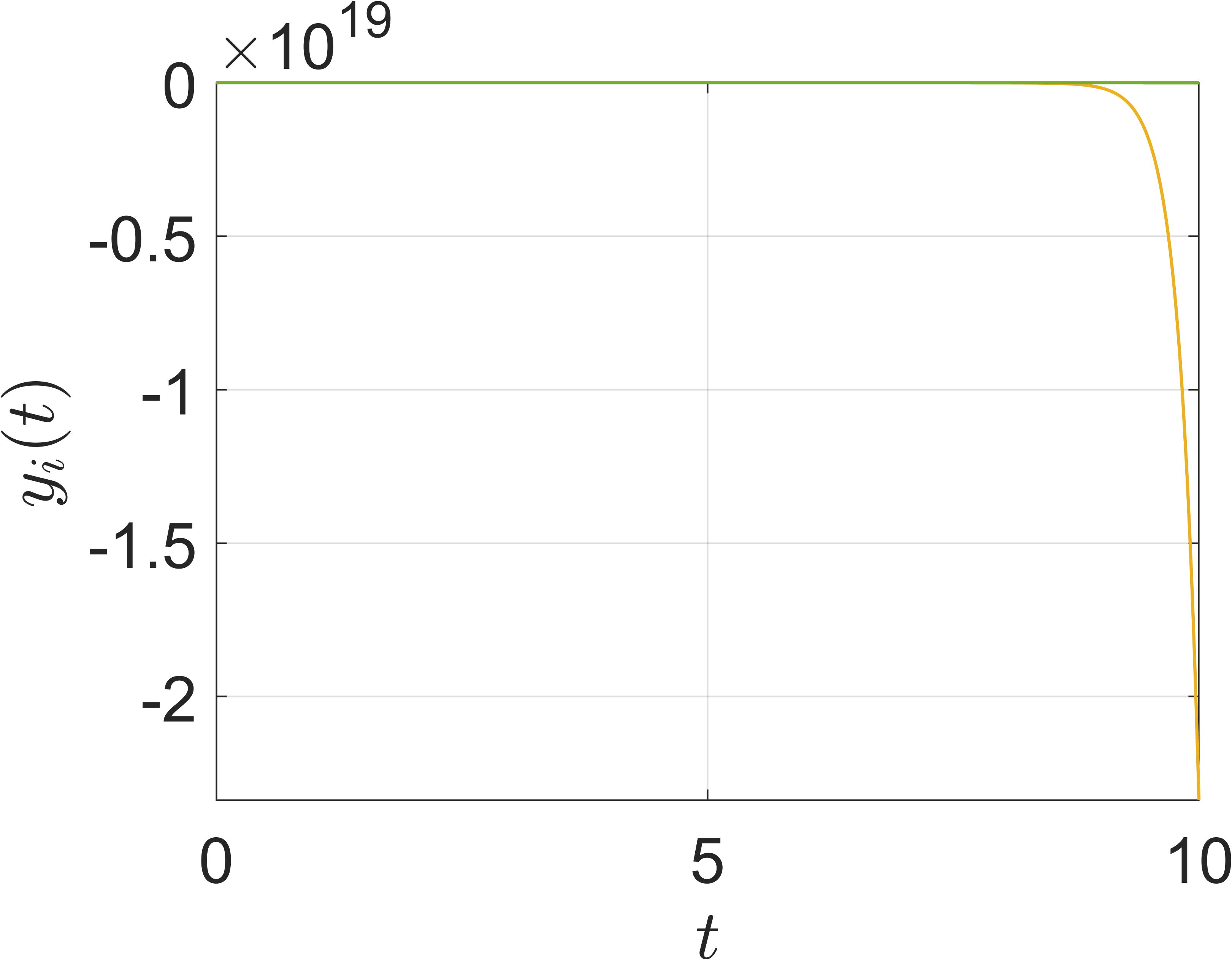}
        \caption{Output with $M=\scriptsize\bm{\0 & \I \\ \I & \0}$.}
        \label{Fig:NSC3Case1y}
    \end{subfigure}
    \begin{subfigure}{0.23\textwidth}
        \includegraphics[width=\linewidth]{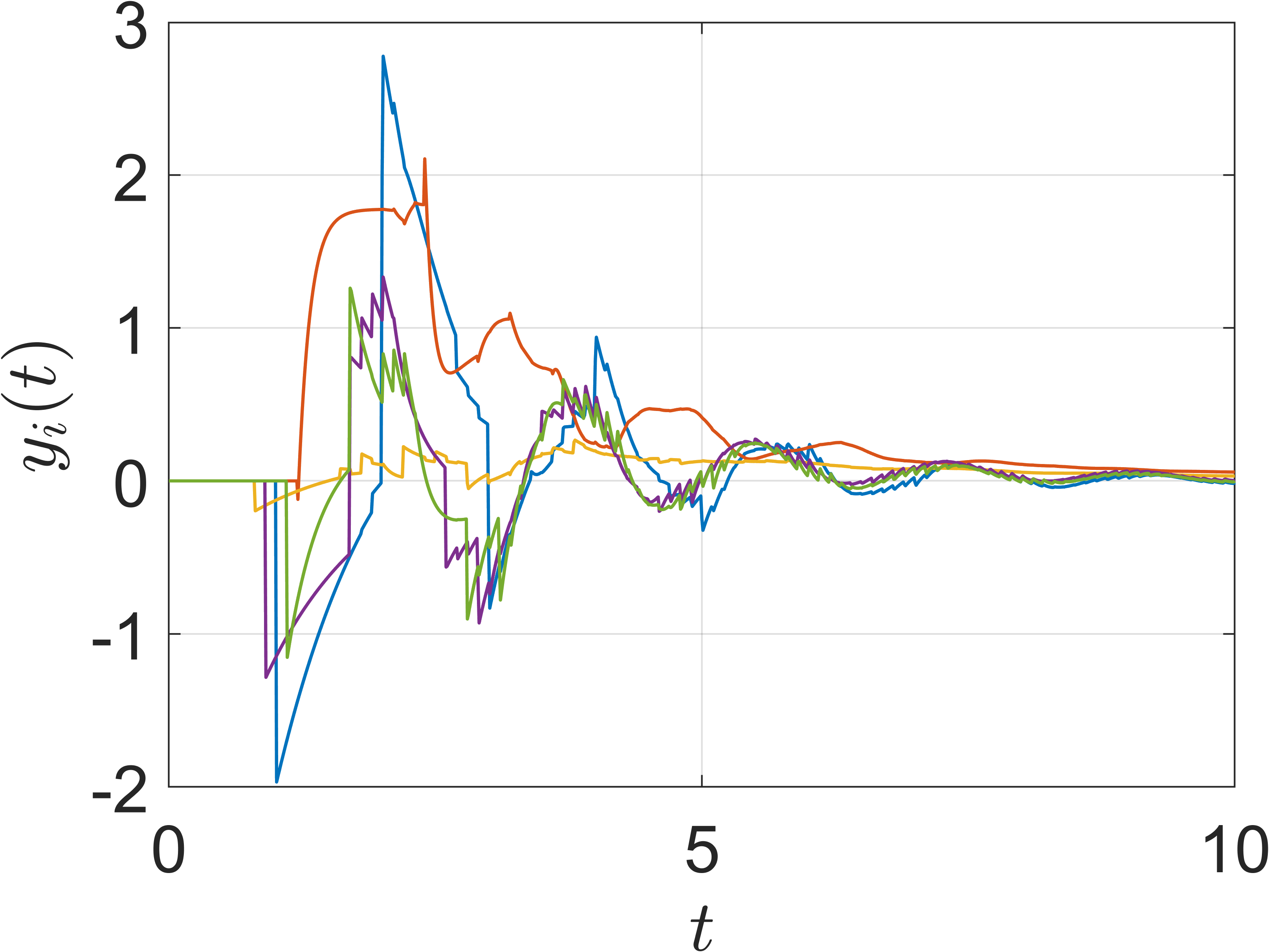}
        \caption{Output $y$ with optimal $M$.}
        \label{Fig:NSC3Case2y}
    \end{subfigure}
    \begin{subfigure}{0.23\textwidth}
        \includegraphics[width=\linewidth]{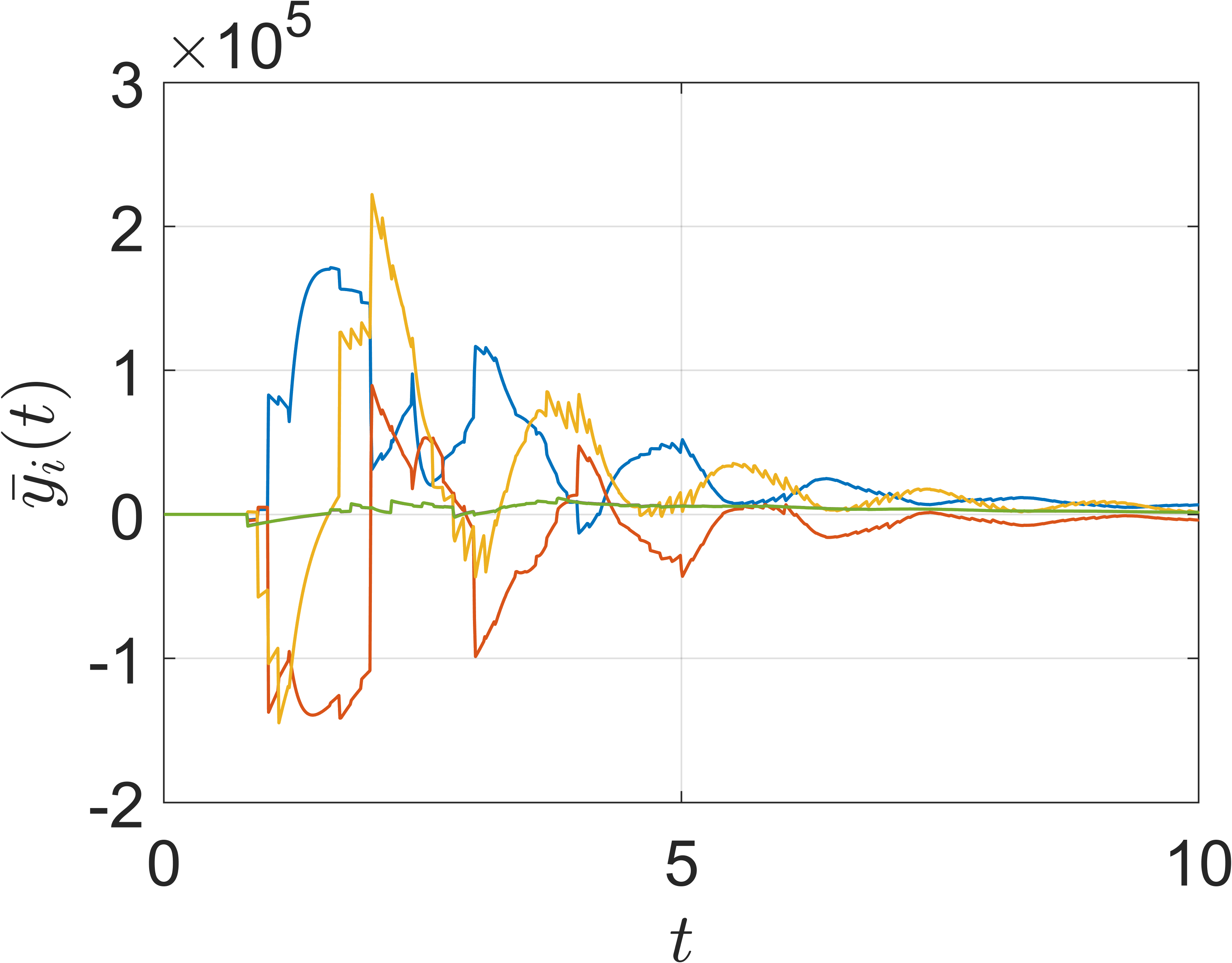}
        \caption{Output $\bar{y}$ with optimal $M$.}
        \label{Fig:NSC3Case2yBar}
    \end{subfigure}
    \caption{\textbf{NSC 3:} \textbf{(a)} An example NSC 3 implemented in Simulink by interconnecting the controllers \eqref{Eq:NumResSubsystems} and the plants \eqref{Eq:NumResSubsystemsBar} via an interconnection matrix $M\in\R^{10 \times 10}$ \eqref{Eq:NSC3Interconnection}. \textbf{(b)} The excitation signal used to purturb the subsystem initial conditions. \textbf{(c)} Observed controller output signal $y(t)$ under $M=\scriptsize\bm{\0 & \I \\ \I & \0}$. Observed: \textbf{(d)} controller and \textbf{(e)} plant output signals under optimally synthesized $M$.
    }
    \label{Fig:NSC3}
\end{figure}

\begin{figure}[!t]
    \centering
    \begin{subfigure}{0.48\textwidth}
    \centering
        \includegraphics[width=0.75\linewidth]{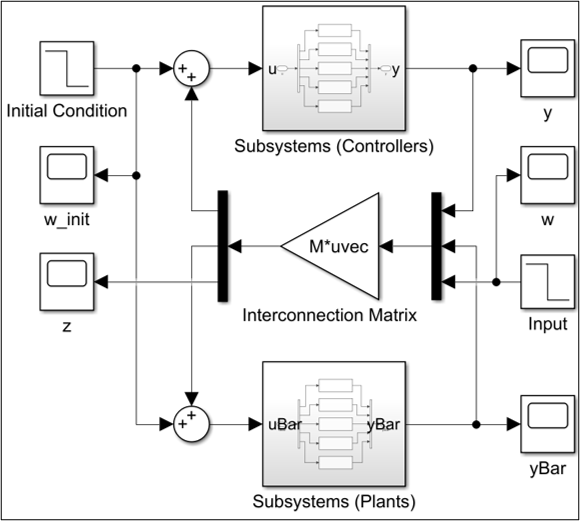}
        \caption{NSC 4: Simulink implementation using the subsystems \eqref{Eq:NumResSubsystems} (as controllers) and \eqref{Eq:NumResSubsystemsBar} (as plants).}
        \label{Fig:NSC4Simulink}
    \end{subfigure}\\
    \begin{subfigure}{0.23\textwidth}
        \includegraphics[width=\linewidth]{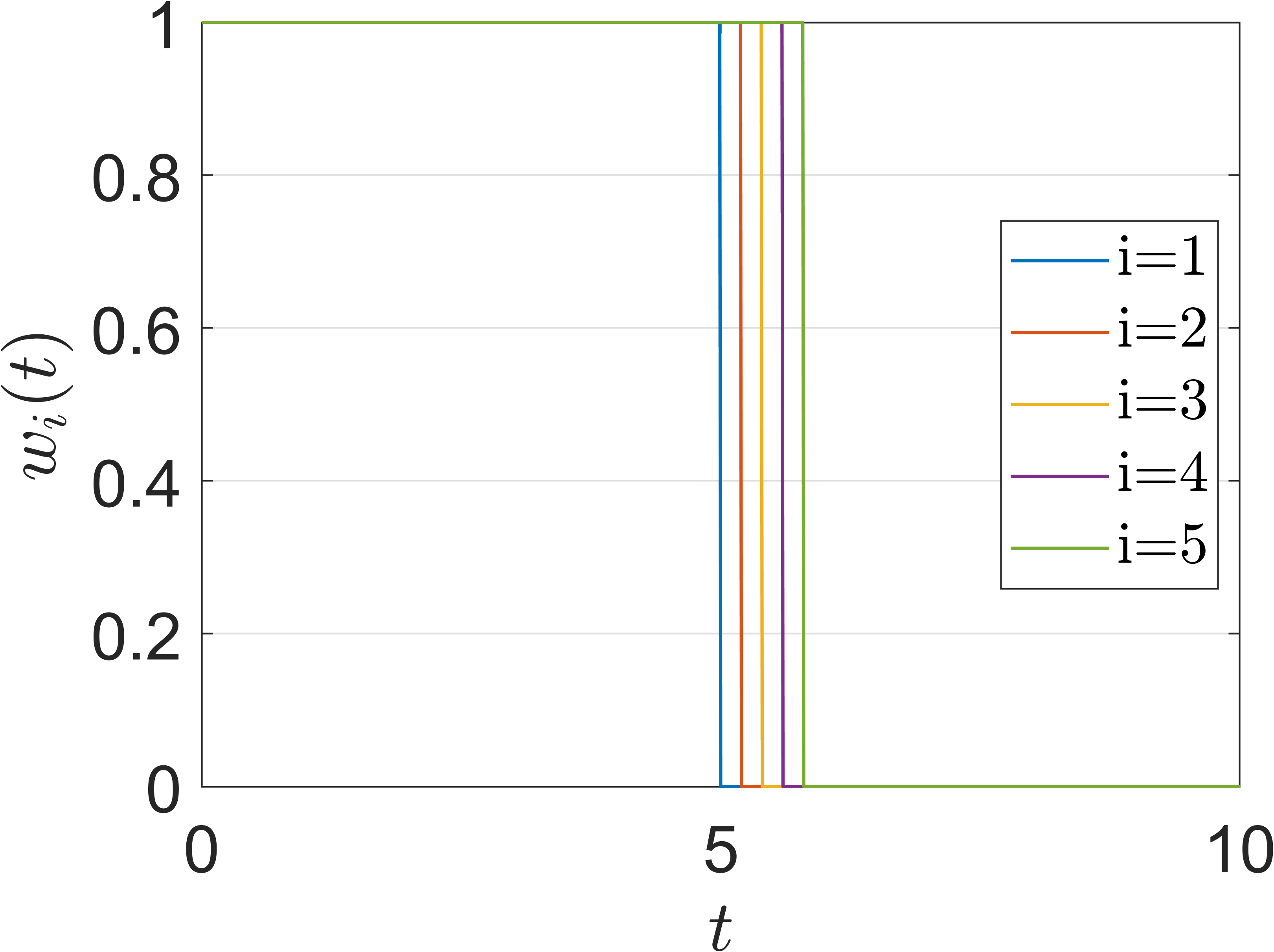}
        \caption{Exogenous input $w(t)$.}
        \label{Fig:NSC4w}
    \end{subfigure}
    \begin{subfigure}{0.23\textwidth}
        \includegraphics[width=\linewidth]{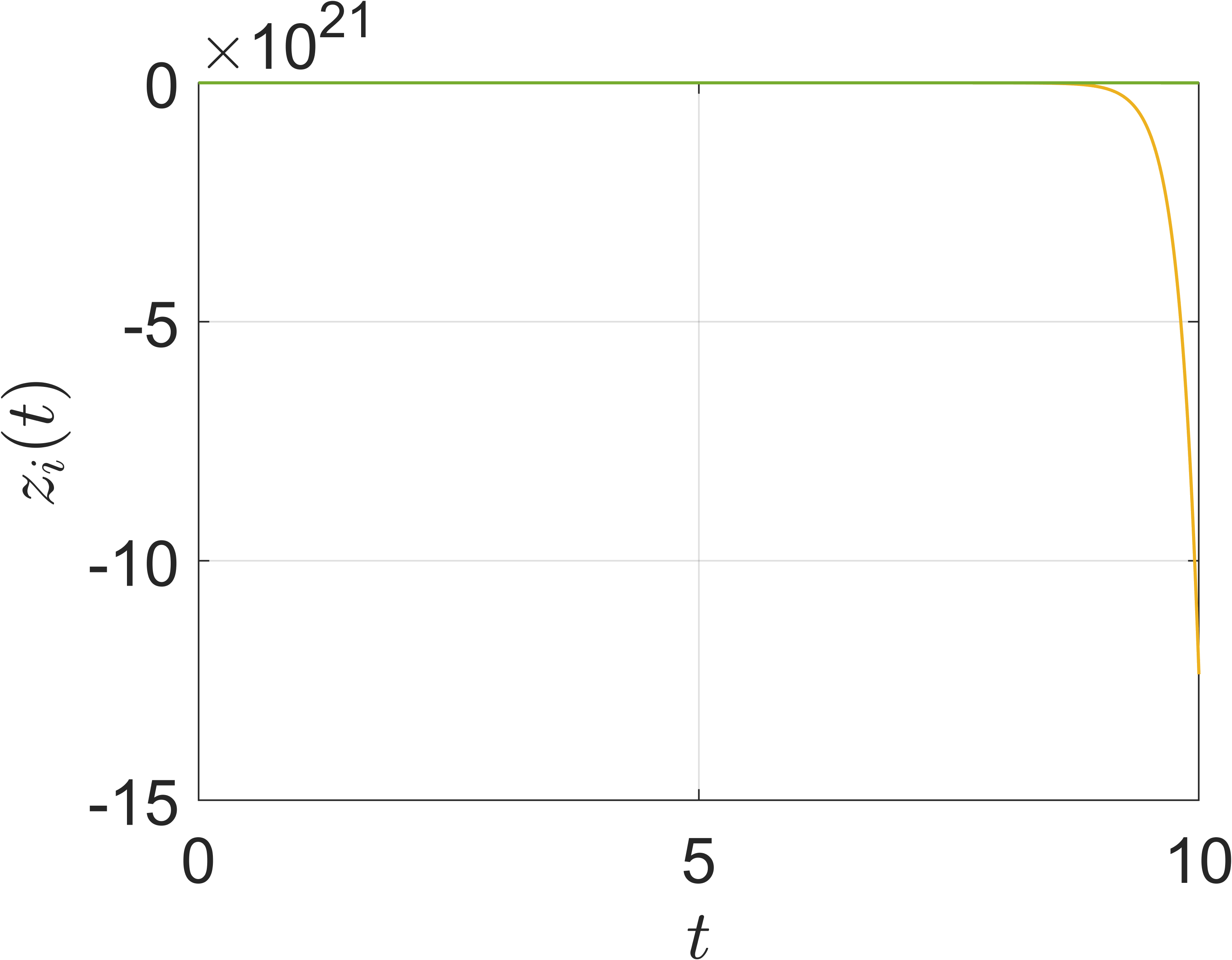}
        \caption{Output with $M=\scriptsize \bm{\0 & \0 & \I \\ \I & \0 & \0 \\ \I & -\I & \0}$.}
        \label{Fig:NSC4Case1z}
    \end{subfigure}
    \begin{subfigure}{0.23\textwidth}
        \includegraphics[width=\linewidth]{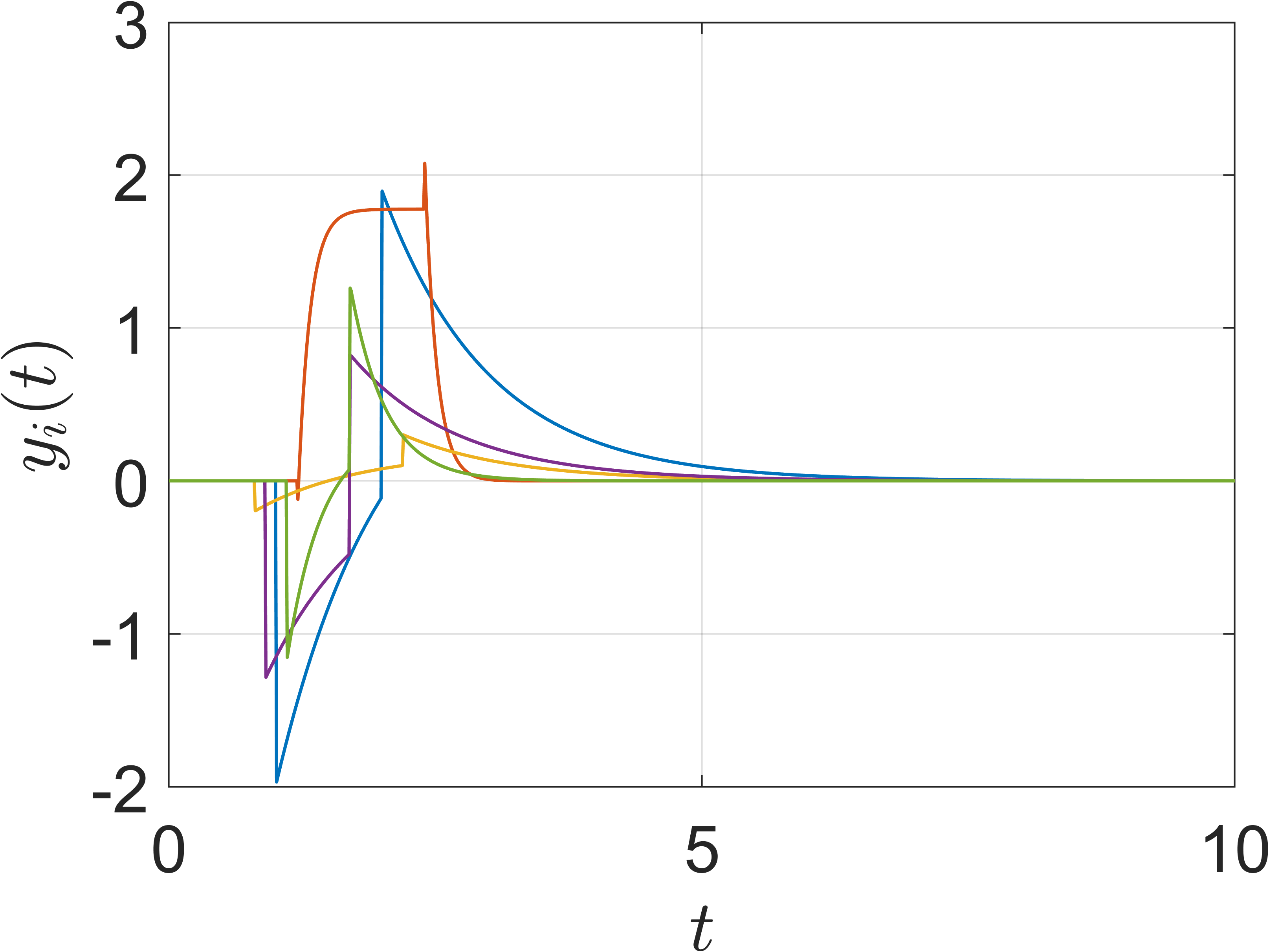}
        \caption{Subsystem outputs with optimal $M$.}
        \label{Fig:NSC4Case2y}
    \end{subfigure}
    \begin{subfigure}{0.23\textwidth}
        \includegraphics[width=\linewidth]{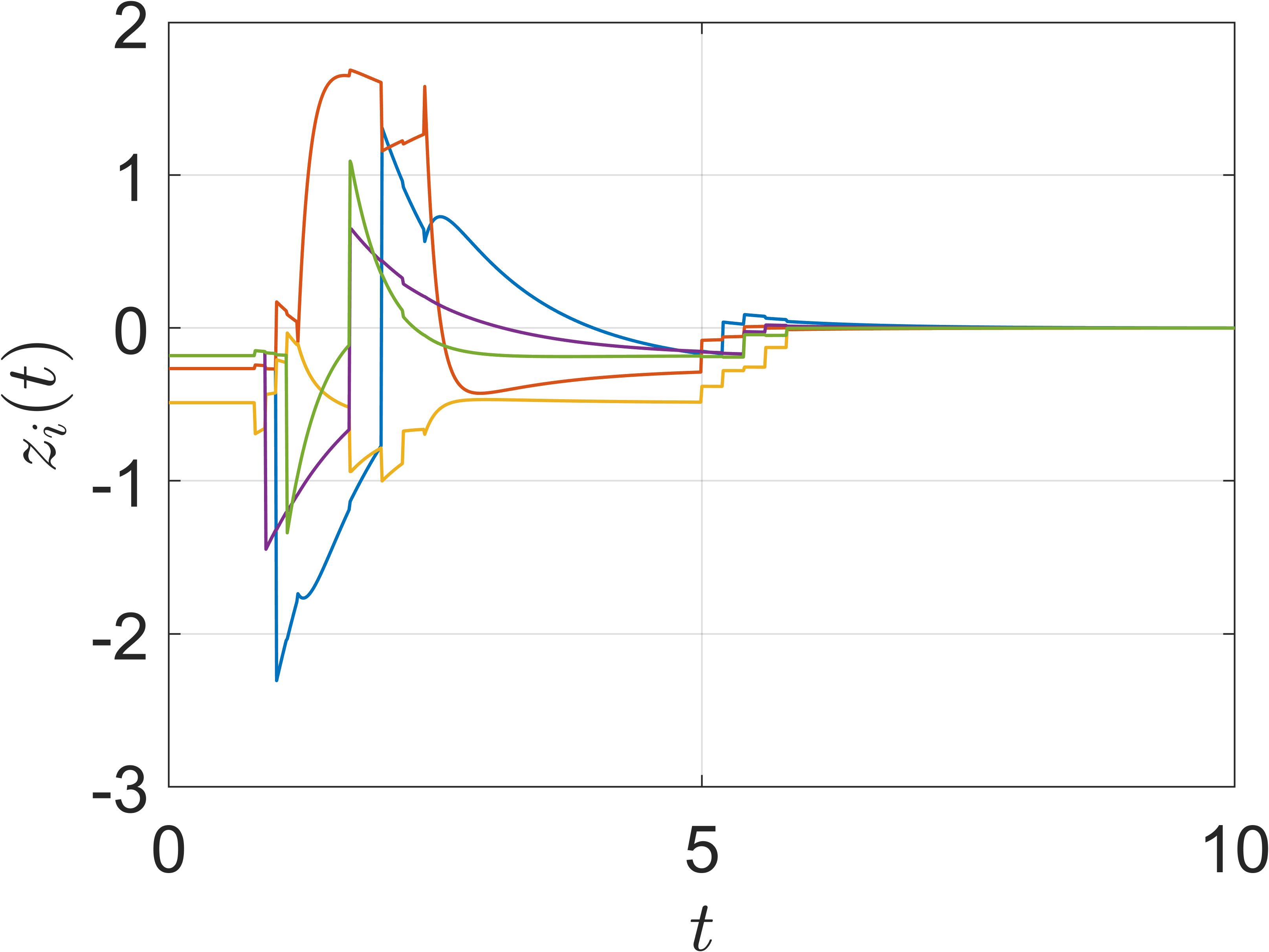}
        \caption{Outputs with optimal $M$.}
        \label{Fig:NSC4Case2z}
    \end{subfigure}
    \caption{\textbf{NSC 4:} \textbf{(a)} An example NSC 4 implemented in matlab using the controllers \eqref{Eq:NumResSubsystems}, plants \eqref{Eq:NumResSubsystemsBar} and an interconnection matrix $M\in\R^{15\times15}$. \textbf{(b)} Used exogenous input signal (in addition to the initial excitation signal shown in Fig. \ref{Fig:NSC3w}). \textbf{(c)} Observed output under $M=\scriptsize \bm{\0 & \0 & \I \\ \I & \0 & \0 \\ \I & -\I & \0}$. 
    Observed: \textbf{(d)} subsystem outputs $y(t)$ and \textbf{(e)} networked system output $z(t)$, under optimally synthesized $M$.}
    \label{Fig:NSC4}
\end{figure}

\subsubsection{\textbf{NSC 3}}
Next, to validate the Prop. \ref{Pr:NSC3Synthesis}, we use the previous networked system (shown in Fig. \ref{Fig:NSC2Simulink}) and connect it to the collection of subsystems (i.e., plants \eqref{Eq:NumResSubsystemsBar}) - making a networked system of the form NSC 3 as shown in Fig. \ref{Fig:NSC3Simulink}. As shown in Fig. \ref{Fig:NSC3Case1y}, this networked system is unstable under the interconnection matrix choice $M= \scriptsize \bm{\I & \I \\ \I & \I}$ \eqref{Eq:NSC3Interconnection}. This motivates the need to synthesize $M$ such that it stabilizes the networked system (i.e., stabilizes the group of unstable plants \eqref{Eq:NumResSubsystemsBar} via the group of controllers \eqref{Eq:NumResSubsystems}). Figures \ref{Fig:NSC3Case2y} and \ref{Fig:NSC3Case2yBar} show the observed output signals when the synthesized $M$ provided by the Prop. \ref{Pr:NSC3Synthesis} is used to interconnect the involved subsystems (i.e., plants and controllers).

\subsubsection{\textbf{NSC 4}}
Finally, to validate the Prop. \ref{Pr:NSC4Synthesis}, we use the previous networked system (shown in Fig. \ref{Fig:NSC3Simulink}) with an added input port ($w$) and an output port ($z$) - making a networked system of the form NSC 4 as shown in Fig. \ref{Fig:NSC4Simulink}. When the interconnection matrix is chosen as $M=\scriptsize \bm{\0 & \0 & \I \\ \I & \0 & \0 \\ \I & -\I & \0}$, as shown in Fig. \ref{Fig:NSC4Case1z}, this networked system amplifies the input signal eventually leading to an unstable behavior. To counter this, we need to synthesize $M$ such that the networked system become finite-gain $L_2$ stable with a gain $\gamma<1$. For this purpose, we used Prop. \eqref{Pr:NSC4Synthesis} and synthesized $M$ that minimizes the $L_2$-gain value of the networked system. Note that, to simplify the problem, we used the ``approximate Simulation'' $M$ matrix format given in Tab. \ref{Tab:Configurations}. The obtained optimal $L_2$-gain value is 
\begin{equation}
    \gamma^* = 0.3351.
\end{equation}
The output trajectories of the resulting finite-gain $L_2$ stable networked system are shown in Figs. \ref{Fig:NSC4Case2y} and \ref{Fig:NSC4Case2z}.

\section{Conclusion}\label{Sec:Conclusion}
In this paper, we considered several widely occurring networked system configurations - each comprised of an interconnected set of non-linear subsystems with known dissipativity properties. For these NSCs, exploiting the known subsystem dissipativity properties, centralized network analysis and network topology synthesis techniques were developed as LMI problems. Then, for these centralized techniques, decentralized counterparts were developed so that they can be executed in a decentralized and compositional manner (i.e., in a way that new/existing subsystems can be added/removed conveniently). Since each proposed technique only uses subsystem dissipativity properties and takes the form of an LMI problem, they can be efficiently implemented and solved. Consequently, the proposed techniques in this paper are ideal for higher-level large-scale networked system design. Finally, we discussed several illustrative numerical examples. Future work aims to extend the proposed techniques for switched and hybrid networked systems.

\bibliographystyle{IEEEtran}
\bibliography{References}

\end{document}